\documentclass[final,onefignum,onetabnum]{siamonline181217}

\usepackage{lipsum}
\usepackage{amsfonts,amssymb}
\usepackage{mathrsfs}
\usepackage{graphicx}
\usepackage{epstopdf}
\usepackage{algorithm}
\usepackage{subfig}
\usepackage{algpseudocode}
\ifpdf
  \DeclareGraphicsExtensions{.eps,.pdf,.png,.jpg}
\else
  \DeclareGraphicsExtensions{.eps}
\fi

\newsiamremark{remark}{Remark}
\newsiamremark{hypothesis}{Hypothesis}
\crefname{hypothesis}{Hypothesis}{Hypotheses}
\newsiamthm{claim}{Claim}

\headers{Gaussian Process Landmarking on Manifolds}{T. Gao, S.Z. Kovalsky, and I. Daubechies}

\title{Gaussian Process Landmarking on Manifolds\thanks{Submitted
\funding{This work is supported by Simons Math+X Investigators Award 400837 and NSF CAREER Award BCS-1552848.}}}

\author{Tingran Gao\thanks{Computational and Applied
    Mathematics Initiative, Department of Statistics, The University of Chicago, Chicago IL (\email{tingrangao@galton.uchicago.edu})}
\and Shahar Z. Kovalsky \thanks{Department of Mathematics, Duke University, Durham, NC (\email{shaharko@math.duke.edu})}
\and Ingrid Daubechies \thanks{Department of Mathematics and Department of Electrical and Computer Engineering, Duke University, Durham NC (\email{ingrid@math.duke.edu})}}

\usepackage{amsopn}

\DeclareMathOperator*{\argmax}{argmax}

\newcommand{\dd}{\,\mathrm{d}}

\newcommand{\GP}{\mathrm{GP}}

\ifpdf
\hypersetup{
  pdftitle={Gaussian Process Landmarking on Manifolds},
  pdfauthor={Tingran Gao, Shahar Z. Kovalsky, and Ingrid Daubechies}
}
\fi

\begin{document}

\maketitle

% REQUIRED
\begin{abstract}
  As a means of improving analysis of biological shapes, we propose an algorithm for sampling a Riemannian manifold by sequentially selecting points with maximum uncertainty under a Gaussian process model. This greedy strategy is known to be near-optimal in the experimental design literature, and appears to outperform the use of user-placed landmarks in representing the geometry of biological objects in our application. In the noiseless regime, we establish an upper bound for the mean squared prediction error (MSPE) in terms of the number of samples and geometric quantities of the manifold, demonstrating that the MSPE for our proposed sequential design decays at a rate comparable to the oracle rate achievable by any sequential or non-sequential optimal design; to our knowledge this is the first result of this type for sequential experimental design. The key is to link the greedy algorithm to reduced basis methods in the context of model reduction for partial differential equations. We expect this approach will find additional applications in other fields of research.
\end{abstract}

% REQUIRED

% REQUIRED

\section{Introduction}
This paper grew out of an attempt to apply principles of the statistics field of \emph{optimal experimental design} to \emph{geometric morphometrics}, a subfield of evolutionary biology that focuses on quantifying the (dis-)similarities between pairs of two-dimensional anatomical surfaces based on their spatial configurations. In contrast to methods for statistical estimation and inference, which typically focus on studying the error made by estimators with respect to a deterministically generated or randomly drawn (but fixed once given) collection of sample observations, and constructing estimators to minimize this error, the paradigm of optimal experimental design is to minimize the empirical risk by an ``optimal'' choice of sample locations, while the estimator itself and the number of samples are kept fixed \cite{Pukelsheim2006,ADT2007}. Finding an optimal design amounts to choosing sample points that are most informative for a class of estimators so as to reduce the number of observations; this is most desirable when even one observation is expensive to acquire (e.g. in spatial analysis (geostatistics) \cite{Stein2012,Cressie2015} and in computationally demanding computer experiments \cite{SWN2013}), but similar ideas have long been exploited in the probabilistic analysis of some classical numerical analysis problems (see e.g. \cite{Smith1918,Ylvisaker1975,Ritter2007}).

In this paper, we adopt the methodology of optimal experimental design for discretely sampling Riemannian manifolds, and propose a greedy algorithm that sequentially selects design points based on the uncertainty modeled by a Gaussian process. On anatomical surfaces of interest to geometric morphometrical applications, these design points play the role of \emph{anatomical landmarks}, or just \emph{landmarks}, which are geometrically or semantically meaningful feature points crafted by evolutionary biologists for quantitatively comparing large collections of biological specimens in the framework of \emph{Procrustes analysis}~\cite{Gower1975,DrydenMardia1998SSA,GowerDijksterhuis2004GPA}. The effectiveness of our approach on anatomical surfaces, along with more background information on geometric morphometrics and Procrustes analysis, is demonstrated in a companion paper \cite{GPLMK2}; though the prototypical application scenario in this paper and \cite{GPLMK2} is geometric morphometrics, we expect the approach proposed here to be more generally applicable to other scientific domains where compact or sparse data representation is demanded. In contexts different from evolutionary biology, closely related (continuous or discretized) manifold sampling problems are addressed in \cite{AS2002,MMRS2004,HS2005}, where  smooth manifolds are  discretized by optimizing the locations of (a fixed number of) points so as to minimize a Riesz functional, or by \cite{oztireli2010,RACT2015}, studying surface simplification via spectral subsampling or geometric relevance. These approaches, when applied to two-dimensional surfaces, tend to distribute points either empirically with respect to fine geometric details preserved in the discretized point clouds or uniformly over the underlying geometric object, whereas triangular meshes encountered in geometric morphometrics often lack fine geometric details but still demand non-uniform, sparse geometric features that are semantically/biologically meaningful; moreover, it is often not clear whether the desired anatomical landmarks are naturally associated with an energy potential. In contrast, our work is inspired by recent research on active learning with Gaussian processes \cite{CGJ1996,RW2006,KGUD2007} as well as related applications in finding landmarks along a manifold \cite{LP2015ICML}. Different from many graph-Laplacian-based manifold landmarking algorithms in semi-supervised learning (e.g. \cite{ZLG2003,XZLD2015}), our approach considers a Gaussian process on the manifold whose covariance structure is specified by the heat kernel, with a greedy landmarking strategy that aims to produce a set of geometrically significant samples with adequate coverage for biological traits. Furthermore, in stark contrast with \cite{dST2004,LF2017} where landmarks are utilized primarily for improving computational efficiency, the landmarks produced by our algorithm explicitly and directly minimize the \emph{mean squared prediction error} (MSPE) and thus bear rich information for machine learning and data mining tasks. The optimality of the proposed greedy procedure is also established (see \Cref{sec:theo-analysis}); this is apparently much less straightforward for non-deterministic, sampling-based manifold landmarking algorithms such as \cite{DLFU2013,CKNS2015,WG2015}.

The rest of this paper is organized as follows. The remainder of this introduction section motivates our main algorithm and discusses other related work. \Cref{sec:background} sets notations and provides background materials for Gaussian processes and the construction of heat kernels on Riemannian manifolds (and discretizations thereof), as well as the ``reweighted kernel'' constructed from these discretized heat kernels; \Cref{sec:gauss-proc-active} presents an unsupervised landmarking algorithm for anatomical surfaces inspired by recent work on uncertainty sampling in Gaussian process active learning \cite{LP2015ICML}; \Cref{sec:theo-analysis} provides the convergence rate analysis and establishes the MSPE optimality; \Cref{sec:disc-future-work} summarizes the current paper with a brief sketch of potential future directions. We defer implementation details of the proposed algorithm for applications in geometric morphometrics to the companion paper \cite{GPLMK2}.

\subsection{Motivation}
\label{sec:motivation}

To see the link between landmark identification and active learning with uncertainty sampling \cite{LG1994,Settles2010}, let us consider the regression problem of estimating a function $f:V\rightarrow\mathbb{R}$ defined over a point cloud $V\subset\mathbb{R}^D$. Rather than construct the estimator from random sample observations, we adopt the point of view of active learning, in which one is allowed to sequentially query the values of $f$ at user-picked vertices $x\in V$. In order to minimize the empirical risk of an estimator $\hat{f}$ within a given number of iterations, the simplest and most commonly used strategy is to first evaluate (under reasonable probabilistic model assumptions) the informativeness of the vertices on the mesh that have not been queried, and then greedily choose to inquire the value of $f$ at the vertex $x$ at which the response value $\hat{f} \left( x \right)$---inferred from all previous queries---is most ``uncertain'' in the sense of attaining highest predictive error (though other uncertainty measures such as the Shannon entropy could be used as well); these sequentially obtained highest-uncertainty points will be treated as morphometrical \emph{landmarks} in our proposed algorithm.

This straightforward application of an active learning strategy summarized above relies upon selecting a regression function $f$ of rich biological information. In the absence of a natural candidate regression function $f$, we seek to reduce in every iteration the maximum ``average uncertainty'' of a class of regression functions, e.g., specified by a Gaussian process prior \cite{RW2006}. Throughout this paper we will denote $\GP \left( m,K \right)$ for the Gaussian process on a smooth, compact Riemannian manifold $M$ with mean function $m:M\rightarrow\mathbb{R}$ and covariance function $K:M\times M\rightarrow\mathbb{R}$. If we interpret choosing a single most ``biologically meaningful'' function $f$ as a manual ``feature handcrafting'' step, the specification of a Gaussian process prior can be viewed as a less restrictive and more stable ``ensemble'' version; the geometric information can be conveniently encoded into the prior by specifying an appropriate covariance function $K$. We construct such a covariance function in \Cref{sec:gauss-proc-regr} by reweighting the heat kernel of the Riemannian manifold $M$, adopting (but in the meanwhile also appending further geometric information to) the methodology of Gaussian process optimal experimental design \cite{SSW1989,SWN2013,FOU2013} and \emph{sensitivity analysis} \cite{ST2002,OOH2004} from the statistics literature.

\subsection{Our Contribution and Other Related Work}
\label{sec:our-contr-other}

The main theoretical contribution of this paper is a convergence rate analysis for the greedy algorithm of uncertainty-based sequential experimental design, which amounts to estimating the uniform rate of decay for the prediction error of a Gaussian process as the number of greedily picked design points increases; on a $C^{\infty}$-manifold we deduce that the convergence is faster than any inverse polynomial rate, which is also the optimal rate any greedy or non-greedy landmarking algorithm can attain on a generic smooth manifold. This analysis makes use of recent results in the analysis of reduced based methods, by converting the Gaussian process experimental design into a basis selection problem in a reproducing kernel Hilbert space associated with the Gaussian process. To our knowledge, there does not exist in the literature any earlier analysis of this type for greedy algorithms in optimal experimental design; the convergence results obtained from this analysis can also be used to bound the number of iterations in Gaussian process active learning \cite{CGJ1996,KGUD2007,LP2015ICML} and maximum entropy design \cite{SWN2013,KSG2008,SKKS2010}. From a numerical linear algebra perspective, though the rank-$1$ update procedure detailed in \cref{rem:rank-one-updates} coincides with the well-known algorithm of pivoted Cholesky decomposition for symmetric positive definite matrices (c.f. \Cref{sec:numer-line-algebra}), we are not aware either of similar results in that context for the performance of pivoting. We thus expect our theoretical contribution to shed light upon gaining deeper understandings of other manifold landmarking algorithms in active and semi-supervised learning \cite{ZLG2003,XZLD2015,CKNS2015,WG2015,dST2004,LF2017}. We discuss implementation details of our algorithm for applications in geometric morphomerics in a companion paper \cite{GPLMK2}.

We point out that, though experimental design is a classical problem in the statistical literature \cite{Fedorov1972,CV1995,Pukelsheim2006}, it is only very recently that interests in computationally efficient experimental design algorithms began to arise in the computer science community \cite{BGS2010,AB2013,Nikolov2015,WYS2017,ZLSW2017,ZLSW2017journal_version}. Most experimental design algorithms, based on various types of optimality criteria including but not limited to A(verage)-, D(eterminant)-, E(igen)-, V(ariance)-, G-optimality and Bayesian alphabetical optimality, are NP-hard computational in their exact form \cite{CM2009,CH2012}, with the only exception of T(race)-optimality which is trivial to solve. For computer scientists, the interesting problem is to find polynomial algorithms that efficiently find $\left( 1+O \left( \epsilon \right) \right)$-approximations of the optimal solution to the exact problem, where $\epsilon>0$ is expected to be as small as possible but depends on the size of the problem and the pre-fixed budget for the number of design points; oftentimes these approximation results also require certain constraints on the relative sizes of the dimension of the ambient space, the number of design points, and the total number of candidate points. Different from those approaches, our theoretical contribution assumes no relations between these quantities, and the convergence rate is with respect to the increasing number of landmark points (as opposed to a pre-fixed budget); nevertheless, similar to results obtained in \cite{BGS2010,AB2013,Nikolov2015,WYS2017,ZLSW2017,ZLSW2017journal_version}, our proposed algorithm has polynomial complexity and is thus computationally tractable; see \Cref{sec:numer-line-algebra} for more details. We refer interested readers to \cite{Pukelsheim2006} for more exhaustive discussions of the optimality criteria used in experimental design.

\section{Background}
\label{sec:background}

\subsection{Heat Kernels and Gaussian Processes on Riemannian Manifolds: A Spectral Embedding Perspective}
\label{sec:spectral-embedding}

Let $\left( M,g \right)$ be an orientable compact Riemannian manifold of dimension $d\geq 1$ with finite volume, where $g$ is the Riemannian metric on $M$. Denote $\dd\mathrm{vol}_M$ for the canonical volume form $M$ with coordinate representation
\begin{equation*}
  \dd\mathrm{vol}_M \left( x \right) = \sqrt{\left| g \left( x \right) \right|}\,\dd x^1\wedge\cdots\wedge\dd x^d.
\end{equation*}
The finite volume will be denoted as
\begin{equation*}
  \mathrm{Vol}\left( M \right)=\int_M\dd\mathrm{vol}_M \left( x \right) = \int_M\sqrt{\left| g \left( x \right) \right|}\,\dd x^1\wedge\cdots\wedge\dd x^d < \infty,
\end{equation*}
and we will fix the canonical normalized volume form $\dd\mathrm{vol}_M/\mathrm{Vol}\left( M \right)$ as reference. Throughout this paper, all distributions on $M$ are absolutely continuous with respect to $\dd\mathrm{vol}_M/\mathrm{Vol}\left( M \right)$.

The single-output regression problem on the Riemannian manifold $M$ will be described as follows. Given independent and identically distributed observations $\left\{ \left(X_i,Y_i\right)\in M\times\mathbb{R}\mid 1\leq i\leq n \right\}$ of a random variable $\left( X,Y \right)$ on the product probability space $M\times \mathbb{R}$, the goal of the regression problem is to estimate the conditional expectation
\begin{equation}
  \label{eq:conditional-expectation}
  f \left( x \right):=\mathbb{E} \left( Y\mid X=x \right)
\end{equation}
which is often referred to as a regression function of $Y$ on $X$ \cite{Tsybakov2008}. The joint distribution of $X$ and $Y$ will always be assumed absolutely continuous with respect to the product measure on $M\times\mathbb{R}$ for simplicity. A Gaussian process (or \emph{Gaussian random field}) on $M$ with mean function $m:M\rightarrow\mathbb{R}$ and covariance function $K:M\times M\rightarrow \mathbb{R}$ is defined as the stochastic process of which any finite marginal distribution on $n$ fixed points $x_1,\cdots,x_n\in M$ is a multivariate Gaussian distribution with mean vector
$$m_n:=\left( m \left( x_1 \right),\cdots,m \left( x_n \right) \right)\in\mathbb{R}^n$$
%$$m \left( \left\{x_1,\cdots,x_n\right\} \right):=\left( m \left( x_1 \right),\cdots,m \left( x_n \right) \right)\in\mathbb{R}^n$$
and covariance matrix
% $$K \left( \left\{ x_1,\cdots,x_n \right\} \right):=
% \begin{pmatrix}
%   K \left( x_1,x_1 \right) & \cdots & K \left( x_1,x_n \right)\\
%   \vdots & & \vdots\\
%   K \left( x_n,x_1 \right) & \cdots & K \left( x_n,x_n \right)
% \end{pmatrix}\in\mathbb{R}^{n\times n}.
% $$
$$K_n :=
\begin{pmatrix}
  K \left( x_1,x_1 \right) & \cdots & K \left( x_1,x_n \right)\\
  \vdots & & \vdots\\
  K \left( x_n,x_1 \right) & \cdots & K \left( x_n,x_n \right)
\end{pmatrix}\in\mathbb{R}^{n\times n}.
$$
A Gaussian process with mean function $m:M\rightarrow\mathbb{R}$ and covariance function $K:M\times M\rightarrow\mathbb{R}$ will be denoted as $\GP \left( m,K \right)$. Under model $Y\sim \GP \left( m,K \right)$, given observed values $y_1,\cdots,y_n$ at locations $x_1,\cdots,x_n$, the \emph{best linear predictor} (BLP) \cite{Stein2012,SWN2013} for the random field at a new point $x$ is given by the conditional expectation
\begin{equation}
\label{eq:blp}
  Y^{*}\left( x \right):=\mathbb{E}\left[ Y \left( x \right)\mid Y \left( x_1 \right)=y_1,\cdots,Y \left( x_n \right)=y_n \right]=m \left( x \right)+k_n \left( x \right)^{\top}K_n^{-1}\left( Y_n-m_n \right)
\end{equation}
where $Y_n=\left( y_1,\cdots,y_n \right)^{\top}\in \mathbb{R}^n$, $k_n \left( x \right)=\left( K \left( x,x_1 \right),\cdots,K \left( x,x_n \right) \right)^{\top}\in \mathbb{R}^n$; at any $x\in M$, the expected squared error, or \emph{mean squared prediction error} (MSPE), is defined as
\begin{equation}
  \label{eq:prediction-error}
  \begin{aligned}
    \mathrm{MSPE}\left( x; x_1,\cdots,x_n \right):&=\mathbb{E}\left[ \left( Y \left( x \right)-Y^{*}\left( x \right) \right)^2 \right]\\
    &=\mathbb{E}\left[ \left(Y \left( x \right)-\mathbb{E}\left[ Y \left( x \right)\mid Y \left( x_1 \right)=y_1,\cdots,Y \left( x_n \right)=y_n \right]\right)^2 \right]\\
    &=K \left( x,x \right)-k_n \left( x \right)^{\top}K_n^{-1}k_n \left( x \right),
  \end{aligned}
\end{equation}
which is a function over $M$. Here the expectation is with respect to all realizations $Y\sim\GP \left( m,K \right)$. Squared integral ($L^2$) or sup ($L^{\infty}$) norms of the pointwise MSPE are often used as a criterion for evaluating the prediction performance over the experimental domain. In geospatial statistics, interpolation with \eqref{eq:blp} and \eqref{eq:prediction-error} is known as \emph{kriging}.%  We will denote
% \begin{equation*}
%   d_n:=\inf_{x_1,\cdots,x_n\in M}\,\sup_{x\in M} \left[ K \left( x,x \right)-k_n \left( x \right)^{\top}K_n^{-1}k_n \left( x \right) \right]
% \end{equation*}

% This paper concerns the comparison between the 

% In this paper we are concerned with the minimum value attainable by optim $x_1,\cdots,x_n$ to minimizing the sup-norm over the manifold $M$ of the function $x\mapsto \mathrm{MSPE}\left( x \right)$, i.e. estimating
% \begin{equation*}
%   \begin{aligned}
%     d_n \left( x \right):=\sup_{x_1,\cdots,x_n\in M}\mathrm{MSPE}\left( x \right)=\sup_{x_1,\cdots,x_n\in M} \left[ K \left( x,x \right)-k_n \left( x \right)^{\top}K_n^{-1}k_n \left( x \right) \right]
% %    &=\sup_{x_1,\cdots,x_n\in M}\mathbb{E}\left[ \left(Y-\mathbb{E}\left[ Y \left( x \right)\mid Y \left( x_1 \right)=y_1,\cdots,Y \left( x_n \right)=y_n \right]\right)^2 \right]
%   \end{aligned}
% \end{equation*}
% and
% \begin{equation*}
%   d_n:=\sup_{x\in M}\,d_n \left( x \right).
% \end{equation*}
%As will be seen in \eqref{eq:kolmogorov-width-gpr}, $d_n$ is the \emph{Kolmogorov width}
%The main theoretical result we provide in this paper is a upper bound for $\sigma_n$ in terms of $n$.
% The current paper focuses on a particular case of \emph{simple} kriging in which $m\equiv 0$, but we emphasize that \eqref{eq:prediction-error} indicates that the prediction error is independent of the function $m:M\rightarrow\mathbb{R}$.
% The special case when $m\equiv 0$ is called \emph{simple} kriging.

Our analysis in this paper concerns the sup norm of the prediction error with $n$ design points $x_1,\cdots,x_n$ picked using a greedy algorithm, i.e. the quantity
\begin{equation*}
  \sigma_n:=\sup_{x\in M}\mathrm{MSPE}\left( x; x_1,\cdots,x_n \right)
%\sup_{x\in M} \left[ K \left( x,x \right)-k_n \left( x \right)^{\top}K_n^{-1}k_n \left( x \right) \right]
\end{equation*}
where $x_1,\cdots,x_n$ are chosen according to \cref{alg:gaussian-process-landmarking}. This quantity is compared with the ``oracle'' prediction error attainable by any sequential or non-sequential experimental design with $n$ points, i.e.
\begin{equation*}
  d_n:=\inf_{x_1,\cdots,x_n\in M}\sup_{x\in M} \mathrm{MSPE}\left( x; x_1,\cdots,x_n \right)
%\left[ K \left( x,x \right)-k_n \left( x \right)^{\top}K_n^{-1}k_n \left( x \right) \right].
\end{equation*}
As will be shown in \eqref{eq:kolmogorov-width-gpr} in Section~\ref{sec:theo-analysis}, $d_n$ can be interpreted as the \emph{Kolmogorov width} of approximating a \emph{Reproducing Kernel Hilbert Space} (RKHS) with a reduced basis. The RKHS we consider is a natural one associated with a Gaussian process; see e.g. \cite{CS2000,MRT2012} for general introductions on RKHS and \cite{vdVvZ2008} for RKHS associated with Gaussian processes. We include a brief sketch of the RKHS theory needed for understanding \Cref{sec:theo-analysis} in \cref{sec:repr-kern-hilb}.

On Riemannian manifolds, there is a natural choice for the kernel function: the heat kernel, i.e. the kernel of the Laplace-Beltrami operator. Denote $\Delta:C^2 \left( M \right)\rightarrow C^2 \left( M \right)$ for the Laplace-Beltrami operator on $M$ with respect to the metric $g$, i.e.
\begin{equation*}
  \Delta f = \frac{1}{\sqrt{\left| g \right|}}\partial_i \left( \sqrt{\left| g \right|}\, g^{ij}\partial_jf \right),\quad \forall f \in C^{\infty}\left( M \right)
\end{equation*}
where the sign convention is such that $-\Delta$ is positive semidefinite. If the manifold $M$ has no boundary, the spectrum of $-\Delta$ is well-known to be real, non-negative, discrete, with eigenvalues satisfying $0=\lambda_0<\lambda_1\leq\lambda_2\leq\cdots\nearrow\infty$, with $\infty$ the only accumulation point of the spectrum; when $M$ has non-empty boundary we assume Dirichlet boundary condition, so the same inequalities hold for the eigenvalues. If we denote $\phi_i$ for the eigenfunction of $\Delta$ corresponding to the eigenvalue $\lambda_i$, then the set $\left\{ \phi_i\mid i=0,1,\cdots \right\}$ constitutes an orthonormal basis for $L^2 \left( M \right)$ under the standard inner product
\begin{equation*}
  \left\langle f_1,f_2 \right\rangle_M:=\int_M f_1 \left( x \right)f_2 \left( x \right)\dd\mathrm{vol}_M \left( x \right).
\end{equation*}
The heat kernel $k_t \left( x,y \right):=k \left( x,y;t \right)\in C^2 \left( M\times M\right)\times C^{\infty}\left(\left( 0,\infty \right)\right)$ is the fundamental solution of the heat equation on $M$:
\begin{equation*}
\partial_tu \left( x,t \right)=-\Delta u \left( x,t \right),\qquad x\in M,\,t\in \left( 0,\infty \right).
\end{equation*}
That is, if the initial data is specified as
\begin{equation*}
u \left( x,t=0 \right)=v \left( x \right)
%  u|_{\partial M}=g
\end{equation*}
then
\begin{equation*}
  u \left( x,t \right) = \int_Mk_t \left( x,y \right)v \left( y \right)\dd\mathrm{vol}_M \left( y \right).
\end{equation*}
In terms of the spectral data of $\Delta$ (see e.g. \cite{Rosenberg1997,BGV2003}), the heat kernel can be written as
\begin{equation}
\label{eq:manifold-heat-kernel}
  k_t \left(x, y\right)=\sum_{i=0}^{\infty}e^{-\lambda_i t}\phi_i \left( x \right)\phi_i\left( y \right),\quad\forall t\geq 0,\,\, x, y\in M.
\end{equation}
For any fixed $t>0$, the heat kernel defines a Mercer kernel on $M$ by
$$ \left( x,y \right)\mapsto k_t \left( x,y \right)\quad\forall \left( x,y \right)\in M\times M$$
and the feature mapping (see \eqref{eq:feature-mapping}) takes the form
\begin{equation}
\label{eq:heat-kernel-expansion-eigenfunctions}
  M\ni x\longmapsto \Phi_t \left( x \right):=\left( e^{-\lambda_0t/2}\phi_0 \left( x \right),e^{-\lambda_1t/2}\phi_1 \left( x \right),\cdots,e^{-\lambda_it/2}\phi_i \left( x \right),\cdots \right)\in\ell^2
\end{equation}
where $\ell^2$ is the infinite sequence space equipped with a standard inner product; see e.g. \cite[\textsection II.1 Example 3]{RS1980}. Note in particular that
\begin{equation}
\label{eq:heat-kernel-as-inner-product}
  k_t \left( x,y \right)=\left\langle \Phi_t \left( x \right),\Phi_t \left( y \right) \right\rangle_{\ell^2}.
\end{equation}
In fact, up to a multiplicative constant $c \left( t \right)=\sqrt{2}\left( 4\pi \right)^{\frac{d}{4}}t^{\frac{n+2}{4}}$, the feature mapping $\Phi_t:M\rightarrow\ell^2$ has long been studied in spectral geometry \cite{BBG1994} and is known to be an embedding of $M$ into $\ell^2$; furthermore, with the multiplicative correction by $c \left( t \right)$, the pullback of the canonical metric on $\ell^2$ is asymptotically equal to the Riemannian metric on $M$.

In this paper we focus on Gaussian processes on Riemannian manifolds with heat kernels (or ``reweighted'' counterparts thereof; see \Cref{sec:gauss-proc-regr}) as covariance functions. There are at least two reasons for heat kernels to be considered as natural candidates for covariance functions of Gaussian processes on manifolds. First, as argued in \cite[\S2.5]{CKP2014}, the abundant geometric information encoded in the Laplace-Beltrami operator makes the heat kernel a canonical choice for Gaussian processes; Gaussian processes defined this way impose natural geometric priors based on randomly rescaled solutions of the heat equation. Second, by \eqref{eq:heat-kernel-as-inner-product}, a Gaussian process on $M$ with heat kernel is equivalent to a Gaussian process on the embedded image of $M$ into $\ell^2$ under the feature mapping \eqref{eq:heat-kernel-expansion-eigenfunctions} with a dot product kernel; this is reminiscent of the methodology of \emph{extrinsic Gaussian process regression} (eGPR) \cite{LNCD2017} on manifolds --- in order to perform Gaussian process regression on a nonlinear manifold, eGPR first embeds the manifold into a Euclidean space using an arbitrary embedding, then performs Gaussian process regression on the embedded image following standard procedures for Gaussian process regression. This spectral embedding interpretation also underlies recent work constructing Gaussian priors, by means of the graph Laplacian, for uncertainty quantification of graph semi-supervised learning \cite{BLSZ2017}.

\subsection{Discretized and Reweighted Heat Kernels}
\label{sec:gauss-proc-regr}

When the Riemannian manifold $M$ is a submanifold embedded in an ambient Euclidean space $\mathbb{R}^D$ ($D\gg d$) and sampled only at finitely many points $\left\{ x_1,\cdots,x_n \right\}$, we know from the literature of Laplacian eigenmaps \cite{LapEigMaps2003,BelkinNiyogi2005} and diffusion maps \cite{CoifmanLafon2006,Singer2006ConvergenceRate,SingerWu2012VDM} that the extrinsic squared exponential kernel matrix
\begin{equation}
\label{eq:sq-exp-kernel-submfld}
  K = \left( K_{ij} \right)_{1\leq i,j\leq n} = \left( \exp \left( -\frac{\left\| x_i-x_j \right\|^2}{t} \right) \right)_{1\leq i,j\leq n}
\end{equation}
is a consistent estimator (up to a multiplicative constant) of the heat kernel of the manifold $M$ if $\left\{ x_i\mid 1\leq i\leq n \right\}$ are sampled uniformly and i.i.d. on $M$ with appropriately adjusted bandwidth parameter $t>0$ as $n\rightarrow\infty$; similar results holds when the squared exponential kernel is replaced with any anisotropic kernel, and additional renormalization techniques can be used to adjust the kernel if the samples are i.i.d. but not uniformly distributed on $M$, see e.g. \cite{CoifmanLafon2006} for more details. These theoretical results in manifold learning justify using extrinsic kernel functions in a Gaussian process regression framework when the manifold is an embedded submanifold of an ambient Euclidean space; the kernel \eqref{eq:sq-exp-kernel-submfld} is also used in \cite{YD2016} for Gaussian process regression on manifolds in a Bayesian setting. Nevertheless, one may well use other discrete approximations of the heat kernel in place of \eqref{eq:sq-exp-kernel-submfld} without affecting our theoretical results in \Cref{sec:theo-analysis}, as long as the kernel matrix $K$ is positive (semi-)definite and defines a valid Gaussian process for our landmarking purposes.

The heat kernel of the Riemannian manifold $M$ defines covariance functions for a family of Gaussian processes on $M$, but this type of covariance functions depends only on the spectral properties of $M$, whereas in practice we would often like to incorporate prior information addressing relative high/low confidence of the selected landmarks. For example, the response variables might be measured with higher accuracy (or equivalently the influence of random observation noise is damped) where the predictor falls on a region on the manifold $M$ with lower curvature. We encode this type of prior information regarding the relative importance of different locations on the domain manifold in a smooth positive weight function $w:M\rightarrow\mathbb{R}_+$ defined on the entire manifold, whereby the higher values of $w \left( x \right)$ indicate a relatively higher importance if a predictor variable is sampled near $x\in M$. Since we assume $M$ is closed, $w$ is bounded below away from zero. To ``knit'' the weight function into the heat kernel, notice that by the reproducing property we have
\begin{equation}
\label{eq:kernel-reproducing}
  k_t \left( x,y \right)=\int_Mk_{t/2}\left( x,z \right)k_{t/2} \left( z,y \right)\dd \mathrm{vol}_M \left( z \right)
\end{equation}
and we can naturally apply the weight function to deform the volume form, i.e. define
\begin{equation}
\label{eq:location-prior-kernel}
  k^w_t \left( x,y \right)=\int_Mk_{t/2}\left( x,z \right)k_{t/2} \left( z,y \right)w \left( z \right)\dd \mathrm{vol}_M \left( z \right).
\end{equation}
Obviously, $k_t^w \left( \cdot,\cdot \right)=k_t \left( \cdot,\cdot \right)$ on $M\times M$ if we pick $w\equiv 1$ on $M$, using the expression \eqref{eq:manifold-heat-kernel} for heat kernel $k_t \left( \cdot,\cdot \right)$ and the orthonormality of the eigenfunctions $\left\{ \phi_i\mid i=0,1,\cdots \right\}$. Intuitively, \eqref{eq:location-prior-kernel} reweighs the mutual interaction between different regions on $M$ such that the portions with high weights have a more significant influence on the covariance structure of the Gaussian process on $M$. Results established for $\GP \left( m,k_t \right)$ can often be directly adapted for $\GP \left( m,k^w_t \right)$.

In practice, when the manifold is sampled only at finitely many i.i.d. points $\left\{ x_1,\cdots,x_n \right\}$ on $M$, the reweighted kernel can be calculated from the discrete extrinsic kernel matrix \eqref{eq:sq-exp-kernel-submfld} with $t$ replaced with $t/2$:
\begin{equation}
  \label{eq:sq-exp-kernel-submfld-reweighted}
    K^w=\left( K^w_{ij} \right)_{1\leq i,j\leq n}=\left( \sum_{k=1}^n e^{-\frac{\left\| x_i-x_k \right\|^2}{t/2}} \cdot w \left( x_k \right)\cdot e^{-\frac{\left\| x_k-x_j \right\|^2}{t/2}} \right)_{1\leq i,j\leq n}=K^{\top}WK
\end{equation}
where $W$ is a diagonal matrix of size $n\times n$ with $w \left( x_k \right)$ at its $k$-th diagonal entry, for all $1\leq k\leq n$, and $K$ is the discrete squared exponential kernel matrix \eqref{eq:sq-exp-kernel-submfld}. It is worth pointing out that the reweighted kernel $K^w$ no longer equals the kernel $K$ in \eqref{eq:sq-exp-kernel-submfld} even when we set $w\equiv 1$ at this discrete level. Similar kernels to \eqref{eq:location-prior-kernel} have also appeared in \cite{CCC2017} as the symmetrization of an asymmetric anisotropic kernel.

Though the reweighting step appears to be a straightforward implementation trick, it turns out to be crucial in the application of automated geometric morphometrics: the landmarking algorithm that will be presented in \Cref{sec:gauss-proc-active} produces biologically much more representative features on anatomical surfaces when the reweighted kernel is adopted. We demonstrate this in greater detail in \cite{GPLMK2}.

\section{Gaussian Process Landmarking}
\label{sec:gauss-proc-active}

We present in this section an algorithm motivated by \cite{LP2015ICML} that automatically places ``landmarks'' on a compact Riemannian manifold using a Gaussian process active learning strategy. Let us begin with an arbitrary nonparametric regression model in the form of \cref{eq:conditional-expectation}. Unlike in standard supervised learning in which a finite number of sample-label pairs are provided, an active learning algorithm can iteratively decide, based on memory of all previously inquired sample-label pairs, which sample to inquire for label in the next step. In other words, given sample-label pairs $\left( X_1,Y_1 \right), \left( X_2,Y_2 \right),\cdots,\left( X_n,Y_n \right)$ observed up to the $n$-th step, an active learning algorithm can decide which sample $X_{n+1}$ to query for the label information $Y_{n+1}=f \left( X_{n+1} \right)$ of the regression function $f$ to be estimated; typically, the algorithm assumes full knowledge of the sample domain, has access to the regression function $f$ as a black box, and strives to optimize its query strategy so as to estimate $f$ in as few steps as possible. With a Gaussian process prior $\GP \left( m,K \right)$ on the regression function class, the joint distribution of a finite collection of $\left( n+1 \right)$ response values $\left( Y_1,\cdots,Y_n,Y_{n+1} \right)$ is assumed to follow a multivariate Gaussian distribution $\mathscr{N}_{n+1} \left( m \left( X_1,\cdots,X_{n+1} \right),K\left( X_1,\cdots,X_{n+1} \right) \right)$ where
\begin{equation}
  \label{eq:gpr-mvn}
  \begin{aligned}
    m \left( X_1,\cdots,X_{n+1} \right) &=
    \begin{pmatrix}
      m \left( X_1 \right),\cdots, m \left( X_{n+1} \right)
    \end{pmatrix}\in\mathbb{R}^n,\\
    K \left( X_1,\cdots,X_{n+1} \right) &=
    \begin{pmatrix}
      K\left( X_1,X_1 \right) & \cdots & K\left( X_1,X_{n+1} \right)\\
      \vdots & &\vdots\\
      K\left( X_{n+1},X_1 \right) & \cdots & K\left( X_{n+1},X_{n+1} \right)
    \end{pmatrix}\in\mathbb{R}^{\left( n+1 \right)\times \left( n+1 \right)}.
  \end{aligned}
\end{equation}
For simplicity of statement, the rest of this paper will use short-hand notations
\begin{equation}
  \label{eq:simplify-notation-vector}
  \begin{aligned}
    &X_n^1=
    \begin{pmatrix}
      X_1,\cdots,X_n
    \end{pmatrix}\in M^n,\quad Y_n^1=
    \begin{pmatrix}
      Y_1,\cdots, Y_n
    \end{pmatrix}\in\mathbb{R}^n,
  \end{aligned}
\end{equation}
and
\begin{equation}
  \label{eq:simplify-notation-matrix}
  \begin{aligned}
    K_{n,n}&=K \left( X_1,\cdots,X_n \right)\in\mathbb{R}^{n\times n},\\
    K \left( X,X_n^1 \right)&=
    \begin{pmatrix}
      K \left( X,X_1 \right),\cdots,K \left( X,X_n \right)
    \end{pmatrix}^{\top}\in\mathbb{R}^n.
  \end{aligned}
\end{equation}
Given $n$ observed samples $\left( X_1,Y_1 \right),\cdots,\left( X_n,Y_n \right)$, at any $X\in M$, the conditional probability of the response value $Y \left( X \right)\mid Y_n^1$ follows a normal distribution $$\mathscr{N} \left( \xi_n \left( X \right), \Sigma_n \left( X \right) \right)$$ where
\begin{equation}
  \label{eq:gaussian-process-new-sample}
  \begin{aligned}
    \xi_n \left( X \right)&=K \left( X,X_n^1 \right)^{\top}K_n^{-1}Y_n^1,\\
    \Sigma_n \left( X \right)&=K \left( X,X \right)-K \left( X,X_n^1 \right)^{\top}K_{n,n}^{-1}K \left( X,X_n^1 \right).
  \end{aligned}
\end{equation}
In our landmarking algorithm, we simply choose $X_{n+1}$ to be the location on the manifold $M$ with the largest variance, i.e.
\begin{equation}
  \label{eq:active-learning-rule}
  \begin{aligned}
    X_{n+1}:=\argmax_{X\in M} \Sigma_n \left( X \right).
%        =&\argmax_{X\in M} \left[K \left( X,X \right)-K \left( X,X_n^1 \right)^{\top}K_{n,n}^{-1}K \left( X,X_n^1 \right)\right].
  \end{aligned}
\end{equation}
Notice that this successive procedure of ``landmarking'' $X_1,X_2,\cdots$ on $M$ is independent of the specific choice of regression function in $\GP \left( m,K \right)$ since we only need the covariance function $K:M\times M\rightarrow\mathbb{R}$.

\subsection{Algorithm}
\label{sec:algorithm}

The main algorithm of this paper, an unsupervised landmarking procedure for anatomical surfaces, will use a discretized, reweighted kernel constructed from triangular meshes that digitize anatomical surfaces. We now describe this algorithm in full detail. Let $M$ be a $2$-dimensional compact surface isometrically embedded in $\mathbb{R}^3$, and denote $\kappa:M\rightarrow\mathbb{R}$, $\eta:M\rightarrow\mathbb{R}$ for the Gaussian curvature and (scalar) mean curvature of $M$. Define a family of weight function $w_{\lambda,\rho}:M\rightarrow\mathbb{R}_{\geq 0}$ parametrized by $\lambda\in \left[ 0,1 \right]$ and $\rho>0$ as
\begin{equation}
  \label{eq:curvature-weight-function}
  w_{\lambda,\rho} \left( x \right) = \frac{\lambda\left| \kappa \left( x \right) \right|^{\rho}}{\displaystyle \int_M \left|\kappa \left( \xi \right)\right|^{\rho}\dd\mathrm{vol}_M \left( \xi \right)}+\frac{\left( 1-\lambda \right)\left| \eta \left( x \right) \right|^{\rho}}{\displaystyle \int_M \left|\eta \left( \xi \right)\right|^{\rho}\dd\mathrm{vol}_M \left( \xi \right)},\quad\forall x\in M.
\end{equation}
This weight function seeks to emphasize the influence of high curvature locations on the surface $M$ on the covariance structure of the Gaussian process prior $\GP \left( m,k_t^{w_{\lambda,\rho}} \right)$, where $k_t^{w_{\lambda,\rho}}$ is the reweighted heat kernel defined in \cref{eq:location-prior-kernel}. We stick in this paper with simple kriging [setting $m\equiv 0$ in $\GP \left( m,K \right)$], and use in our implementation default values $\lambda=1/2$ and $\rho=1$ (but one may wish to alter these values to fine-tune the landscape of the weight function for a specific application).

%but occasionally alter the parameters $\lambda$ and $\rho$ to fine-tune the landscape of the weight function (nevertheless, unless otherwise specified, we set by default $\lambda=1/2$ and $\rho=1$).

For all practical purposes, we only concern ourselves with $M$ being a piecewise linear surface, represented as a discrete triangular mesh $T=\left( V,E \right)$ with vertex set $V = \left\{ x_1,\cdots,x_{\left| V \right|} \right\}\subset\mathbb{R}^3$ and edge set $E$. We calculate the mean and Gaussian curvature functions $\eta,\kappa$ on the triangular mesh $\left( V,E \right)$ using standard algorithms in computational geometry \cite{CM2003,ACDLD2003}. The weight function $w_{\lambda,\rho}$ can then be calculated at each vertex $x_i$ by
\begin{equation}
\label{eq:weight-function-discrete}
  w_{\lambda,\rho}\left( x_i \right)=\frac{\lambda \left| \kappa \left( x_i \right) \right|^{\rho}}{\displaystyle \sum_{k=1}^{\left| V \right|} \left| \kappa \left( x_k \right) \right|^{\rho}\nu \left( x_k \right)}+\frac{\left( 1-\lambda \right) \left| \eta \left( x_i \right) \right|^{\rho}}{\displaystyle \sum_{k=1}^{\left| V \right|} \left| \eta \left( x_k \right) \right|^{\rho}\nu \left( x_k \right)},\qquad \forall x_i\in V
\end{equation}
where $\nu \left( x_k \right)$ is the area of the Voronoi cell of the triangular mesh $T$ centered at $x_i$. The reweighted heat kernel $k_t^{w_{\lambda,\rho}}$ is then defined on $V\times V$ as
\begin{equation}
\label{eq:location-prior-kernel-discrete}
  k_t^{w_{\lambda,\rho}}\left( x_i,x_j \right)=\sum_{k=1}^{\left| V \right|}k_{t/2}\left( x_i,x_k \right)k_{t/2}\left( x_k,x_j \right)w_{\lambda,\rho}\left( x_k \right)\nu \left( x_k \right)
\end{equation}
where the (unweighted) heat kernel $k_t$ is calculated as in \cref{eq:sq-exp-kernel-submfld}. Until a fixed total number of landmarks are collected, at step $\left( n+1 \right)$ the algorithm computes the uncertainty score $\Sigma_{\left( n+1 \right)}$ on $V$ from the existing $n$ landmarks $\xi_1,\cdots,\xi_n$  by
\begin{equation}
\label{eq:uncertainty-score}
  \Sigma_{\left( n+1 \right)} \left( x_i \right)=k_t^{w_{\lambda,\rho}}\left( x_i, x_i \right)-
  k_t^{w_{\lambda,\rho}} \left( x_i,\xi_n^1 \right)^{\top}
  k_t^{w_{\lambda,\rho}}\left( \xi_n^1,\xi_n^1 \right)^{-1}k_t^{w_{\lambda,\rho}} \left( x_i,\xi_n^1 \right)\quad\forall x_i\in V
\end{equation}
where
\begin{equation*}
  \begin{aligned}
    k_t^{w_{\lambda,\rho}} \left( x_i,\xi_n^1 \right)&:=\begin{pmatrix}
    k_t^{w_{\lambda,\rho}}\left( x_i,\xi_1 \right)\\
    \vdots\\
    k_t^{w_{\lambda,\rho}}\left( x_i,\xi_n \right)
  \end{pmatrix},\\
  k_t^{w_{\lambda,\rho}} \left( \xi_n^1,\xi_n^1 \right)&:=\begin{pmatrix}
    k_t^{w_{\lambda,\rho}}\left( \xi_1,\xi_1 \right) & \cdots & k_t^{w_{\lambda,\rho}}\left( \xi_1,\xi_n \right)\\
    \vdots & & \vdots\\
    k_t^{w_{\lambda,\rho}}\left( \xi_n,\xi_1 \right) & \cdots & k_t^{w_{\lambda,\rho}}\left( \xi_n,\xi_n \right)
  \end{pmatrix},
  \end{aligned}
\end{equation*}
and pick the $\left( n+1 \right)$-th landmark $\xi_{n+1}$ according to the rule
\begin{equation*}
  \xi_{n+1} = \argmax_{x_i\in V}\Sigma_{\left( n+1 \right)}\left( x_i \right).
\end{equation*}
If there are more than one maximizer of $\Sigma_{\left( n+1 \right)}$, we just randomly pick one; at step $1$ the algorithm simply picks the vertex maximizing $x\mapsto k_t^{w_{\lambda,\rho}}\left( x,x \right)$ on $V$. See \cref{alg:gaussian-process-landmarking} for a comprehensive description.

\begin{algorithm}
\caption{{\sc Gaussian Process Landmarking with Reweighted Heat Kernel}}
\label{alg:gaussian-process-landmarking}
 \begin{algorithmic}[1]
   \Procedure{GPL}{$T$, $L$, $\lambda\in \left[ 0,1 \right]$, $\rho>0$, $\epsilon>0$}\Comment{Triangular Mesh $T =\left( V,E \right)$, number of landmarks $L$}
   \State{$\kappa,\eta\gets$ {\sc DiscreteCurvatures}$\left( T \right)$}\Comment{calculate discrete Gaussian curvature $\kappa$ and mean curvature $\eta$ on $T$}
   \State{$\nu\gets$ {\sc VoronoiAreas}$\left( T \right)$}\Comment{calculate the area of Voronoi cells around each vertex $x_i$}
   \State{$w_{\lambda,\rho}\gets$ {\sc CalculateWeight}$\left(\kappa,\eta,\lambda,\rho,\nu\right)$}\Comment{calculate weight function $w_{\lambda,\rho}$ according to \cref{eq:weight-function-discrete}}
   \State{$W\gets\left[ \exp\left(-\left\| x_i-x_j \right\|^2\bigg/\epsilon\right) \right]_{1\leq i,j\leq \left| V \right|}\in\mathbb{R}^{\left| V \right|\times \left| V \right|}$}
   \State{$\Lambda\gets\mathrm{diag}\left( w_{\lambda,\rho}\left( x_1 \right)\nu \left( x_1 \right), \cdots, w_{\lambda,\rho}\left( x_{\left| V \right|} \right)\nu \left( x_{\left| V \right|} \right)\right)\in\mathbb{R}^{\left| V \right|\times \left| V \right|}$}
   \State{$\xi_1,\cdots,\xi_L\gets\emptyset$}\Comment{initialize landmark list}
   \State{$\Psi\gets 0$}
   \State{$\ell\gets 1$}
   \State{$K_{\mathrm{full}}\gets W^{\top}\Lambda W\in\mathbb{R}^{\left| V \right|\times \left| V \right|}$}
   \State{$K_{\mathrm{trace}}\gets \mathrm{diag}\left(K_{\mathrm{full}}\right)\in\mathbb{R}^{\left| V \right|}$}
   \While{$\ell < L+1$}
     \If {$\ell=1$}
       \State $\Sigma\gets K_{\mathrm{trace}}$
     \Else \Comment{calculate uncertainty scores by \cref{eq:uncertainty-score}}
       \State $b\gets$ {solve linear system} $\Psi \left[\left[\xi_1,\cdots,\xi_{\ell}\right],: \right]b = \Psi$
       \State $\Sigma\gets K_{\mathrm{trace}}-\textrm{diag}\left(\Psi^{\top} b\right)\in\mathbb{R}^{\left| V \right|}$
     \EndIf
     \State $\displaystyle\xi_{\ell}\gets \argmax\Sigma$
     \State $\Psi\gets K_{\mathrm{full}} \left[ :,\left[\xi_1,\cdots,\xi_{\ell}\right] \right]$
     \State $\ell\gets \ell+1$
   \EndWhile
 \State \Return $\xi_1,\cdots,\xi_L$
   \EndProcedure
 \end{algorithmic}
\end{algorithm}

\begin{remark}
  \cref{alg:gaussian-process-landmarking} can be easily adapted to work with point clouds (where connectivity information is not present) and in higher dimensional spaces, which makes it applicable to a much wider range of input data in geometric morphometrics as well as other applications; see e.g. \cite{GPLMK2}. For instance, it suffices to replace Step 4 of \cref{alg:gaussian-process-landmarking} with a different discrete curvature (or another type of ``importance score'') calculation procedure on point clouds (see e.g. \cite{Rusu_ICRA2011_PCL,CLMT2015}), and replace Step 5 with a nearest-neighbor weighted graph adjacency matrix construction. We require in this paper the inputs to be triangular meshes with edge connectivity only for the ease of statement, as computation of discrete curvatures on triangular meshes is much more straightforward.
\end{remark}

\begin{remark}
  \label{rem:rank-one-updates}
  Note that, according to \cref{eq:uncertainty-score}, each step adds only one new row and one new column to the inverse covariance matrix, which enables us to perform rank-$1$ updates to the covariance matrix according to the block matrix inversion formula (see e.g. \cite[\S A.3]{RW2006})
  \begin{equation*}
    K_n^{-1}=
    \begin{pmatrix}
      K_{n-1} & P\\
      P^{\top} & K \left( X_n,X_n \right)
    \end{pmatrix}^{-1}
    =
    \begin{pmatrix}
      K_{n-1}^{-1}\left( I_{n-1}+\mu PP^{\top}K_{n-1}^{-1} \right) & -\mu K_{n-1}^{-1}P\\
      -\mu P^{\top} K_{n-1}^{-1} & \mu
    \end{pmatrix}
  \end{equation*}
where
\begin{equation*}
  \begin{aligned}
    P &=
  \begin{pmatrix}
    K \left( X_1,X_n \right),\cdots,K \left( X_{n-1},X_n \right)
  \end{pmatrix}\in\mathbb{R}^{n-1},\\
  \mu &= \left( K \left( X_n,X_n \right)-P^{\top}K_{n-1}^{-1}P \right)^{-1}\in\mathbb{R}.
  \end{aligned}
\end{equation*}
This simple trick significantly improves the computational efficiency as it avoids directly inverting the covariance matrix when the number of landmarks becomes large as the iteration progresses.
\end{remark}

Before we delve into the theoretical aspects of \cref{alg:gaussian-process-landmarking}, let us present a few typical instances of this algorithm in practical use. A more comprehensive evaluation of the applicability of \cref{alg:gaussian-process-landmarking} to geometric morphometrics is deferred to \cite{GPLMK2}. In a nutshell, the Gaussian process landmarking algorithm picks the landmarks on the triangular mesh successively, according to the uncertainty score function $\Sigma$ at the beginning of each step; at the end of each step the uncertainty score function gets updated, with the information of the newly picked landmark incorporated into the inverse convariance matrix defined as in \cref{eq:gaussian-process-new-sample}. \cref{fig:B03-Lmk-Progression} illustrates the first few successive steps on a triangular mesh discretization of a fossil molar of primate \emph{Plesiadapoidea}. Empirically, we observed that the updates on the uncertainty score function are mostly local, i.e. no abrupt changes of the uncertainty score are observed away from a small geodesic neighborhood centered at each new landmark. Guided by uncertainty and curvature-reweighted covariance function, the Gaussian process landmarking often identifies landmarks of abundant biological information---for instance, the first Gaussian process landmarks are often highly biologically informative, and demonstrate comparable level of coverage with observer landmarks manually picked by human experts. See \cref{fig:Q10-ObLmk-vs-GPLmk} for a visual comparison between the automatically generated landmarks with the observer landmarks manually placed by evolutionary anthropologists on a different digitized fossil molar.

\begin{figure}[htp]
  \centering
  \includegraphics[width=1.0\textwidth]{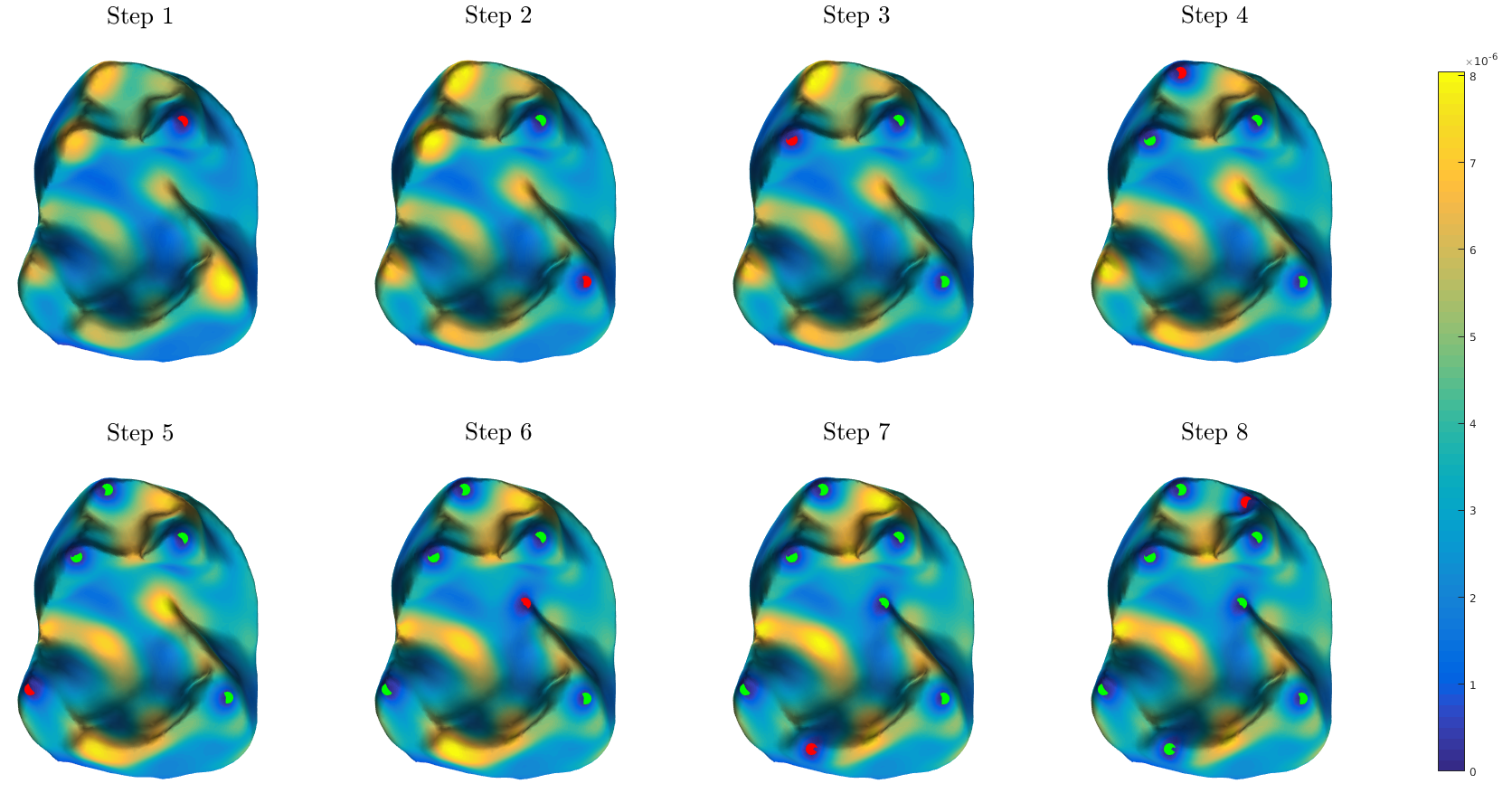}
  \caption{\small The first $8$ landmarks picked successively by Gaussian Process Landmarking (\cref{alg:gaussian-process-landmarking}) on a digitized fossil molar of \emph{Plesiadapoidea} (extinct mammals from the Paleocene and Eocene of North America, Europe, and Asia \cite{Plesiadapiform2017}), with the uncertainty scores at the end of each step rendered on the triangular mesh as a heat map. In each subfigure, the pre-existing landmarks are colored green, and the new landmark is colored red. At each step, the algorithm picks the vertex on the triangular mesh with the highest uncertainty score (computed according to \cref{eq:gaussian-process-new-sample}), then updates the score function.}
  \label{fig:B03-Lmk-Progression}
\end{figure}

\begin{figure}[htb]
  \centering
  \includegraphics[width=1.0\textwidth]{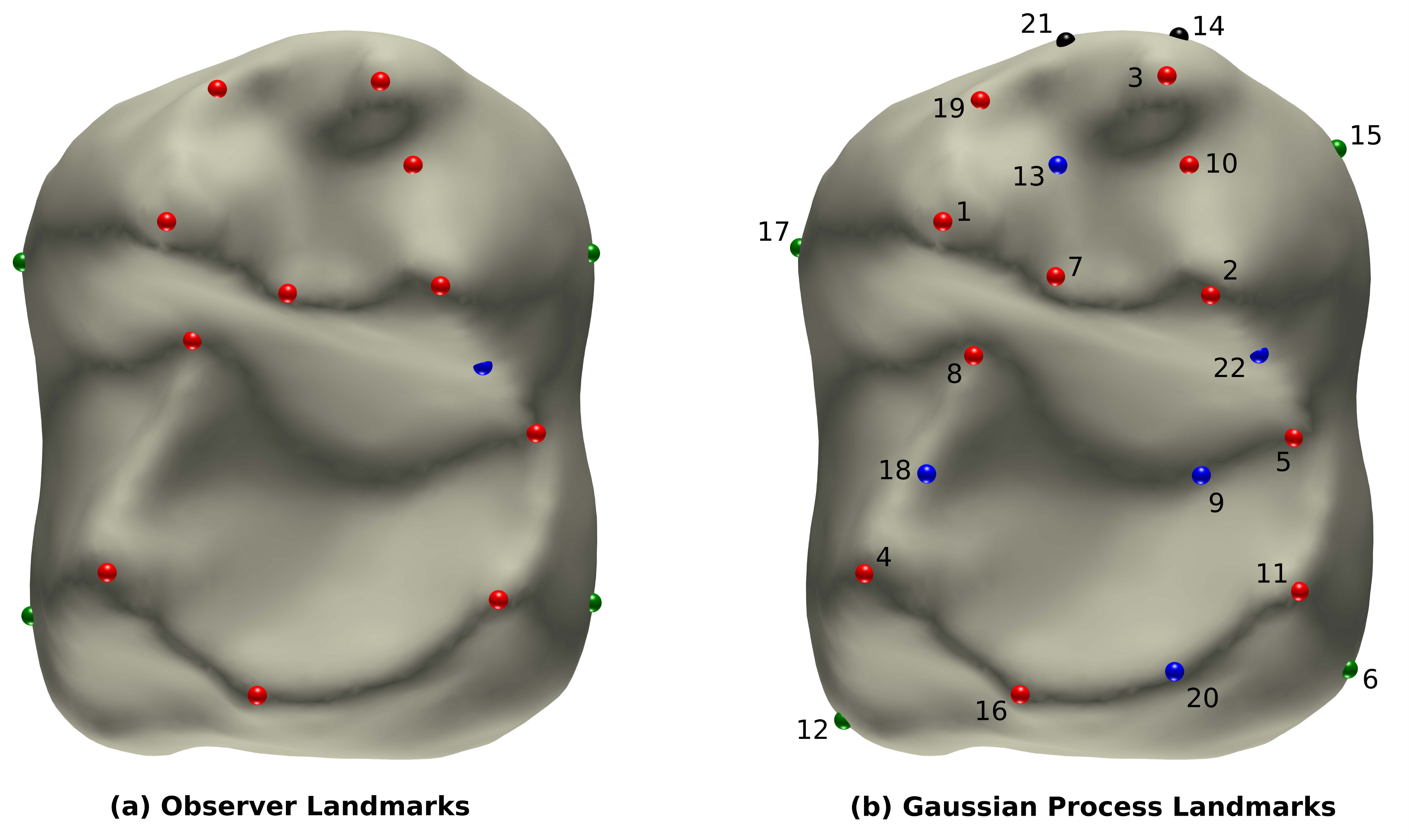}
  \caption{\small\textbf{Left:} Sixteen observer landmarks on a digitized fossil molar of a \emph{Teilhardina} (one of the oldest known fossil primates closely related with living tarsiers and anthropoids \cite{Teilhardina2017}) identified manually by evolutionary anthropologists as ground truth, first published in \cite{PNAS2011}. \textbf{Right:} The first $22$ landmarks picked by Gaussian Process Landmarking (\cref{alg:gaussian-process-landmarking}). The numbers next to each landmark indicate the order of appearance. The Gaussian process landmarks strikingly resembles the observer landmarks: the red landmarks (Number $1$-$5$, $7$, $8$, $10$, $11$, $16$, $19$) signal geometric sharp features (cusps or saddle points corresponding to local maximum/minimum Gaussian curvature); the blue landmarks sit either along the curvy cusp ridges and grooves (Number $13$, $18$, $20$, $22$) or at the basin (Number $9$), serving the role often played by \emph{semilandmarks} (c.f. \cite[\S2.1]{GPLMK2}); the four green landmarks (Number $6$, $12$, $15$, $17$) approximately delimit the ``outline'' of the tooth in occlusal view.}
  \label{fig:Q10-ObLmk-vs-GPLmk}
\end{figure}

\subsection{Numerical Linear Algebra Perspective}
\label{sec:numer-line-algebra}

\cref{alg:gaussian-process-landmarking} can be divided into $2$ phases: Line 1 to 10 focus on constructing the kernel matrix $K_{\mathrm{full}}$ from the geometry of the triangular mesh $M$; from Line 11 onward, only numerical linear algebraic manipulations are involved. In fact, the numerical linear algebra part of \cref{alg:gaussian-process-landmarking} is identical to Gaussian elimination (or LU decomposition) with a very particular ``diagonal pivoting'' strategy, which is different from standard full or partial pivoting in Gaussian elimination. To see this, first note that the variance $\Sigma_n \left( X \right)$ in \cref{eq:gaussian-process-new-sample} is just the diagonal of the Schur complement of the $n\times n$ submatrix of $K_{\mathrm{full}}$ corresponding to the $n$ previously chosen landmarks $X_1,\cdots,X_n$, and recall from \cite[Ex. 20.3]{TrefethenBau1997} that this Schur complement arises as the bottom-right $\left(\left| V \right|-n\right)\times \left(\left| V \right|-n\right)$ block after the $n$-th elimination step. The greedy criterion \cref{eq:active-learning-rule} then amounts to selecting the largest diagonal entry in this Schur complement as the next pivot. Therefore, the second phase of \cref{alg:gaussian-process-landmarking} can be consolidated into the form of a ``diagonal-pivoted'' LU decomposition, i.e. $K_{\mathrm{full}}P=LU$, in which the first $L$ columns of the permutation matrix $P$ reveals the location of the $L$ chosen landmarks. In fact, since the kernel matrix we choose is symmetric and positive semidefinite, the rank-$1$ updates in \cref{rem:rank-one-updates} most closely resembles the \emph{pivoted Cholesky decomposition} (see e.g. \cite[\S10.3]{Higham2002} or \cite{HPS2012}). Identical with these classical pivoting-based matrix decomposition algorithms, the time and space computational complexity of the main algorithm in \cref{alg:gaussian-process-landmarking} are thus $O \left( L^3 \right)$ and $O \left( n^2 \right)$, respectively, where $n$ is the total number of candidate points and $L$ is the desired number of parameters. Note that these complexities are both polynomial and comparable with those in the computer science literature \cite{BGS2010,Nikolov2015,WYS2017,ZLSW2017,ZLSW2017journal_version}. This numerical linear algebraic perspective motivates us to investigate variants of \cref{alg:gaussian-process-landmarking} based on other numerical linear algebraic algorithms with pivoting in future work.

\section{Rate of Convergence: Reduced Basis Methods in Reproducing Kernel Hilbert Spaces}
\label{sec:theo-analysis}

In this subsection we analyze the rate of convergence of our main Gaussian process landmarking algorithm in Section~\ref{sec:gauss-proc-active}. While the notion of ``convergence rate'' in the context of Gaussian process regression (i.e. kriging \cite{Molnar1985,Stein2012}) or scattered data approximation (see e.g. \cite{Wendland2004} and the references therein) refers to how fast the interpolant approaches the true function, our focus in this paper is the rate of convergence of \cref{alg:gaussian-process-landmarking} itself, i.e. the number of steps the algorithm takes before it terminates. In practice, unless a maximum number of landmarks is predetermined, a natural criterion for terminating the algorithm is to specify a threshold for the sup-norm of the prediction error \eqref{eq:prediction-error} [i.e. the variance \eqref{eq:active-learning-rule}] over the manifold. We emphasize again that, although this greedy approach is motivated by the probabilistic model of Gaussian processes, the MSPE is completely determined once the kernel function and the design points are specified, and so is the greedy algorithmic procedure. Our analysis is centered around bounding the uniform rate at which the pointwise MSPE function \eqref{eq:prediction-error} decays with respect to number landmarks greedily selected.

% In this subsection we analyze the rate of convergence of our main Gaussian process landmarking algorithm in \Cref{sec:gauss-proc-active}. While the notion of ``convergence rate'' in the context of Gaussian process regression (i.e. \emph{kriging} \cite{Molnar1985,Stein2012}) or scattered data approximation (see e.g. \cite{Wendland2004} and the references therein) refers to how fast the interpolant approaches the true function, our focus in this paper is the rate of convergence of \cref{alg:gaussian-process-landmarking} itself, for instance, the number of steps the algorithm takes before it terminates. In practice, unless a maximum number of landmarks is predetermined, a natural criterion for terminating the algorithm is to specify a threshold for the $\infty$-norm of the conditional covariance \cref{eq:active-learning-rule} over the manifold. We emphasize again that, although this greedy approach is motivated by the probabilistic model of Gaussian processes, the conditional covariance function as well as the greedy algorithm of manifold landmarking is completely deterministic once the kernel function is known. Our analysis is thus centered around bounding the uniform rate at which the conditional covariance decays with respect to number landmarks greedily selected.

To this end, we observe the connection between \cref{alg:gaussian-process-landmarking} and a greedy algorithm studied thoroughly for \emph{reduced basis methods} in \cite{BCDDvPW2011,DvPW2013} in the context of model reduction. While the analyses in \cite{BCDDvPW2011,DvPW2013} assume general Hilbert and Banach spaces, we apply their result to a reproducing kernel Hilbert space (RKHS), denoted as $\mathscr{H}_K$, naturally associated with a Gaussian process $\GP \left( m,K \right)$; as will be demonstrated below, the MSPE with respect to $n$ selected landmarks can be interpreted as a distance function between elements of $\mathscr{H}_K$ to an $n$-dimensional subspace of $\mathscr{H}_K$ determined by the selected landmarks. We emphasize that, though the connection between Gaussian process and RKHS is well known (see e.g. \cite{vdVvZ2008} and the references therein), we are not aware of existing literature addressing the resemblance between the two classes of greedy algorithms widely used in Gaussian process experimental design and reduced basis methods.

%To this end, we observe the connection between \cref{alg:gaussian-process-landmarking} and a greedy algorithm studied thoroughly for \emph{reduced basis methods} in \cite{BCDDvPW2011,DvPW2013} in the context of model reduction. While the analyses in \cite{BCDDvPW2011,DvPW2013} assume general Hilbert and Banach spaces, we apply their result to a reproducing kernel Hilbert space (RKHS), denoted as $\mathscr{H}_K$, naturally associated with a Gaussian process $\GP \left( m,K \right)$; as will be demonstrated below, the conditional covariance function with respect to $n$ selected landmarks can be interpreted as a distance function between elements of $\mathscr{H}_K$ to an $n$-dimensional subspace of $\mathscr{H}_K$ determined by the selected landmarks. We emphasize that, though the connection between Gaussian process and RKHS is well known (see e.g. \cite{vdVvZ2008} and the references therein), we are not aware of existing literature addressing the resemblance between the two classes of greedy algorithms widely used in Gaussian process experimental design and reduced basis methods.

We begin with a brief summary of the greedy algorithm in reduced basis methods for a general Banach space $\left( X,\left\| \cdot \right\| \right)$. The algorithm strives to approximate \emph{all} elements of $X$ using a properly constructed linear subspace spanned by (as few as possible) selected elements from a compact subset $\mathscr{F}\subset X$; thus the name ``reduced'' basis. A popular greedy algorithm for this purpose generates successive approximation spaces by choosing the first basis $f_1\in\mathscr{F}$ according to
\begin{equation}
\label{eq:reduced-basis-initial}
  f_1:=\argmax_{f\in\mathscr{F}}\left\| f \right\|
\end{equation}
and, successively, when $f_1,\cdots,f_{n-1}$ are picked already, choose
\begin{equation}
\label{eq:reduced-basis-induction}
  f_{n+1}:=\argmax_{f\in\mathscr{F}}\mathrm{dist}\left( f,V_n \right)
\end{equation}
where
\begin{equation*}
  V_n=\mathrm{span}\left\{ f_1,f_2,\cdots,f_n \right\}
\end{equation*}
and
\begin{equation*}
  \mathrm{dist}\left( f,V_n \right):=\inf_{g\in V_n}\left\| f-g \right\|.
\end{equation*}
In words, at each step we greedily pick the function that is ``farthest away'' from the set of already chosen basis elements. Intuitively, this is analogous to the \emph{farthest point sampling} (FPS) algorithm \cite{Gonzalez1985,MoenningDodgson2003} in Banach spaces, with a key difference in the choice of the distance between a point $p$ and a set of selected points $\left\{ q_1,\cdots q_n \right\}$: in FPS such a distance is defined as the maximum over all distances $\left\{\left\| p-q_i \right\|\mid 1\leq i\leq n\right\}$, whereas in reduced basis methods the distance is between $p$ and the linear subspace spanned by $\left\{ q_1,\cdots,q_n \right\}$.

Gaussian process landmarking algorithm fits naturally into the framework of reduced basis methods, as follows. Let us first specialize this construction to the case when $X$ is the reproducing kernel Hilbert space $\mathscr{H}_K\subset L^2 \left( M \right)$, where $M$ is a compact Riemannian manifold and $K$ is the reproducing kernel. A natural choice for $K$ is the heat kernel $k_t \left( \cdot,\cdot \right):M\times M\rightarrow\mathbb{R}$ with a fixed $t>0$ as in \Cref{sec:spectral-embedding}, but for a submanifold isometrically embedded into an ambient Euclidean space it is common as well to choose the kernel to be the restriction to $M$ of a positive (semi-)definite kernel in the ambient Euclidean space such as \cref{eq:sq-exp-kernel-submfld} or \cref{eq:sq-exp-kernel-submfld-reweighted}, for which Sobolev-type error estimates are known in the literature of scattered data approximation \cite{NWW2006,FW2012}. It follows from standard RKHS theory (see e.g. \eqref{eq:defn-rkhs-intro-reinterp}) that
\begin{equation}
\label{eq;defn-rkhs-manifold}
  \mathscr{H}_K=\overline{\mathrm{span}\left\{ \sum_{i\in I}a_i K \left( \cdot,x_i \right)\mid a_i\in\mathbb{R},x_i\in M,\mathrm{card}\left( I \right)<\infty \right\}}
\end{equation}
and, by the compactness of $M$ and the regularity of the kernel function, we have for any $x\in M$
\begin{equation*}
  \left\langle K \left( \cdot,x \right),K \left( \cdot,x \right) \right\rangle_{\mathscr{H}_K}=K \left( x,x \right)\leq \left\| K \right\|_{\infty,M\times M}<\infty
\end{equation*}
which justifies the compactness of
\begin{equation}
\label{eq:compact-reduced-basis-set-our-setting}
  \mathscr{F}:=\mathrm{span}\left\{ K \left( \cdot,x \right)\mid x\in M \right\}
\end{equation}
as a subset of $\mathscr{H}_K$, since $\mathscr{H}_K$ embeds into $L^2 \left( M \right)$ compactly \cite{Aronszajn1950,RBV2010}. In fact, since we only used the compactness of $M$ and the boundedness of $K$ on $M\times M$, the argument above for the compactness of $\mathscr{F}$ can be extended to any Gaussian process defined on a compact metric space with a bounded kernel. The initialization step \cref{eq:reduced-basis-initial} now amounts to selecting $K \left( \cdot,x \right)$ from $\mathscr{F}$ that maximizes
$$\left\| K \left( \cdot,x \right) \right\|_{\mathscr{H}_K}^2=\left\langle K \left( \cdot,x \right),K \left( \cdot,x \right) \right\rangle_{\mathscr{H}_K}=K \left( x,x \right)$$
which is identical to \cref{eq:active-learning-rule} when $n=1$ (or equivalently, Line 14 in \cref{alg:gaussian-process-landmarking}); furthermore, given $n\geq 1$ previously selected basis functions $K \left( \cdot,x_1 \right),\cdots,K \left( \cdot,x_n \right)$, the $\left(n+1\right)$-th basis function will be chosen according to \cref{eq:reduced-basis-induction}, i.e. $f_{n+1}=K \left( \cdot,x_n \right)$ maximizes the infimum
\begin{align}
%\label{eq:maximizing-infimum}
%  \begin{aligned}
    \inf_{g\in \mathrm{span}\left\{ K \left( \cdot,x_1 \right),\cdots,K \left( \cdot,x_n \right) \right\}} &\left\| K \left( \cdot,x \right)-g \right\|_{\mathscr{H}_K}^2=\inf_{a_1,\cdots,a_n\in\mathbb{R}} \left\| K \left( \cdot,x \right)-\sum_{i=1}^na_i K \left( \cdot,x_i \right) \right\|_{\mathscr{H}_K}^2\nonumber\\
    =&\inf_{a_1,\cdots,a_n\in\mathbb{R}} K \left( x,x \right)-2\sum_{i=1}^na_iK \left( x,x_i \right)+\sum_{i=1}^n\sum_{j=1}^na_ia_j K \left( x_i,x_j \right)\nonumber\\
   \stackrel{\left( * \right)}{=}&K \left( x,x \right)-K \left( x,x_n^1 \right)^{\top}
    K_{n,n}^{-1}
    K \left( x,x_n^1 \right)\label{eq:maximizing-infimum}
%  \end{aligned}
\end{align}
where the notation are as in \cref{eq:simplify-notation-vector} and \cref{eq:simplify-notation-matrix}, i.e.
\begin{equation*}
  K \left( x,x_n^1 \right):=\begin{pmatrix}
      K \left( x,x_1 \right)\\
      \vdots\\
      K \left( x,x_n \right)
    \end{pmatrix},\quad K_{n,n}:=\begin{pmatrix}
      K \left( x_1,x_1 \right)&\cdots &K \left( x_1,x_n \right)\\
      \vdots & & \vdots\\
      K \left( x_n,x_1 \right)&\cdots&K \left( x_n,x_n \right)
    \end{pmatrix}.
\end{equation*}
The equality $\left( * \right)$ follows from the observation that, for any fixed $x\in M$, the minimizing vector $\mathbf{a}:=\left( a_1,\cdots,a_n \right)^{\top}\in\mathbb{R}^n$ satisfies
\begin{equation*}
  \begin{aligned}
    K \left( x,x_n^1 \right)&=
  K_{n,n}\mathbf{a} \quad \Leftrightarrow\quad 
  \mathbf{a} =
  K_{n,n}^{-1}
   K \left( x,x_n^1 \right).
  \end{aligned}
\end{equation*}
It is clear at this point that maximizing the rightmost quantity in \cref{eq:maximizing-infimum} is equivalent to following the greedy landmark selection criterion \cref{eq:active-learning-rule} at the $\left( n+1 \right)$-th step. We thus conclude that \cref{alg:gaussian-process-landmarking} is equivalent to the greedy algorithm for reduced basis method in $\mathscr{H}_K$, a reproducing kernel Hilbert space modeled on the compact manifold $M$. The following lemma summarizes this observation for future reference.
\begin{lemma}
\label{lem:basic-project-estimates}
  Let $M$ be a compact Riemannian manifold, and let $K:M\times M\rightarrow\mathbb{R}$ be a positive semidefinite kernel function. Consider the reproducing kernel Hilbert space $\mathscr{H}_K\subset L^2 \left( M \right)$ as defined in \cref{eq;defn-rkhs-manifold}. For any $x\in M$ and a collection of $n$ points $\mathscr{X}_n=\left\{x_1,x_2,\cdots,x_n\right\}\subset M$, the orthogonal projection $P_n$ from $\mathscr{H}_K$ to $V_n=\mathrm{span}\left\{ K \left( \cdot,x_i \right)\mid 1\leq i\leq n \right\}$ is
  \begin{equation*}
    P_n\left(K \left( \cdot,x \right)\right) = \sum_{i=1}^na_i^{*}\left( x \right)K \left( \cdot,x_i \right)
  \end{equation*}
where $a_i^{*}:M\rightarrow \mathbb{R}$ is the inner product of vector $\left( K \left( x,x_1 \right),\cdots,K \left( x,x_n \right) \right)$ with the $i$-th row of
\begin{equation*}
  \begin{pmatrix}
      K \left( x_1,x_1 \right)&\cdots &K \left( x_1,x_n \right)\\
      \vdots & & \vdots\\
      K \left( x_n,x_1 \right)&\cdots&K \left( x_n,x_n \right)
   \end{pmatrix}^{-1}.
\end{equation*}
In particular, $a_i^{*}$ has the same regularity as the kernel $\Phi$, for all $1\leq i\leq n$. Moreover, the squared distance between $K \left( \cdot,x \right)$ and the linear subspace $V_n\subset\mathscr{H}_K$ has the closed-form expression
\begin{equation}
\label{eq:dist-between-pt-to-plane}
\begin{aligned}
   P_{K,\mathscr{X}_n}\left( x \right):&=\left\| K \left( \cdot,x \right)-P_n \left( K \left( \cdot,x \right) \right) \right\|^2_{\mathscr{H}_K}\\
   &=\min_{a_1,\cdots,a_n\in\mathbb{R}}\left\| K \left( \cdot,x \right)-\sum_{i=1}^na_i K \left( \cdot,x_i \right) \right\|_{\mathscr{H}_K}^2\\
   &= K \left( x,x \right)-K \left( x,x_n^1 \right)^{\top}
    \begin{pmatrix}
      K \left( x_1,x_1 \right)&\cdots &K \left( x_1,x_n \right)\\
      \vdots & & \vdots\\
      K \left( x_n,x_1 \right)&\cdots&K \left( x_n,x_n \right)
    \end{pmatrix}^{-1}
    K \left( x,x_n^1 \right)
\end{aligned}
\end{equation}
where
\begin{equation*}
  K \left( x,x_n^1 \right):=\left( K \left( x,x_1 \right), \cdots, K \left( x,x_n \right) \right)^{\top}\in\mathbb{R}^n.
\end{equation*}
Consequently, for any Gaussian process defined on $M$ with covariance structure given by the kernel function $K$, the MSPE of the Gaussian process conditioned on the observations at $x_1,\cdots,x_n\in M$ equals to the distance between $K \left( \cdot,x \right)$ and the subspace $V_n$ spanned by $K \left( \cdot,x_1 \right),\cdots,K \left( \cdot,x_n \right)$.
%Consequently, for any Gaussian process defined on $M$ with covariance structure given by the kernel function $K$, the conditional covariance of the Gaussian process conditioned on the observations at $x_1,\cdots,x_n\in M$ equals to the distance between $K \left( \cdot,x \right)$ and the subspace $V_n$ spanned by $K \left( \cdot,x_1 \right),\cdots,K \left( \cdot,x_n \right)$.
\end{lemma}
%\begin{remark}
The function $P_{K,\mathscr{X}_n}:M\rightarrow\mathbb{R}_{\geq 0}$ defined in \cref{eq:dist-between-pt-to-plane} is in fact the squared \emph{power function} in the literature of scattered data approximation; see e.g. \cite[Definition 11.2]{Wendland2004}.
%  The function $P_{K,\mathscr{X}_n}$  is  \emph{power function} in the literature of scattered data approximation (c.f. \cite[]{Wendland2004}).
%\end{remark}

% \begin{remark}
%   \TG{Explain the connection with FPS in the ``Gaussian Process variance distance''?}
% \end{remark}

The convergence rate of greedy algorithms for reduced basis methods has been investigated in a series of works \cite{BMPPT2012,BCDDvPW2011,DvPW2013}. The general paradigm is to compare the maximum approximation error incurring after the $n$-th greedy step, denoted as
\begin{equation}
\label{eq:defn-sigma}
  \sigma_n:=\mathrm{dist}\left( f_{n+1},V_n \right)=\max_{f\in\mathscr{F}}\mathrm{dist}\left( f,V_n \right),
\end{equation}
with the \emph{Kolmogorov width} (c.f. \cite{LvGM1996}), a quantity characterizing the theoretical optimal error of approximation using any $n$-dimensional linear subspace generated from any greedy or non-greedy algorithms, defined as
\begin{equation}
\label{eq:defn-dn}
  d_n:=\inf_Y\sup_{f\in \mathscr{F}}\mathrm{dist}\left( f,Y \right)
\end{equation}
where the first infimum is taken over all $n$-dimensional subspaces $Y$ of $X$. When $n=1$, both $\sigma_1$ and $d_1$ reduce to the $\infty$-bound of the kernel function on $M\times M$, i.e. $\left\| K \right\|_{\infty,M\times M}$. Note that by definitions \eqref{eq:defn-sigma} \eqref{eq:defn-dn} both sequences $\left\{ \sigma_n\mid n\in\mathbb{N} \right\}$ and $\left\{ d_n\mid n\in\mathbb{N} \right\}$ are monotonically non-decreasing, since $V_1\subsetneqq V_2 \subsetneqq\cdots$; see also \cite[\textsection 1.3]{BCDDvPW2011}. In \cite{DvPW2013} the following comparison between $\left\{ \sigma_n\mid n\in\mathbb{N} \right\}$ and $\left\{ d_n\mid n\in\mathbb{N} \right\}$ was established:
\begin{theorem}[\cite{DvPW2013}, Theorem 3.2 (The $\gamma=1$ Case)]
\label{eq:greedy-reduced-basis-theorem}
  For any $N\geq 0$, $n\geq 1$, and $1 \leq m < n$, there holds
  \begin{equation*}
    \prod_{\ell=1}^n\sigma_{N+\ell}^2\leq \left( \frac{n}{m} \right)^m \left( \frac{n}{n-m} \right)^{n-m}\sigma_{N+1}^{2m}d_m^{2n-2m}.
  \end{equation*}
\end{theorem}
% This result establishes a non-asymptotic optimality of the greedy basis selection procedure.
This result can be used to establish a direct comparison between the performance of greedy and optimal basis selection procedures. For instance, setting $N=0$ and taking advantage of the monotonicity of the sequence $\left\{ \sigma_n\mid n\in\mathbb{N} \right\}$, one has from \cref{eq:greedy-reduced-basis-theorem} that
\begin{equation*}
  \sigma_n\leq \sqrt{2}\min_{1\leq m < n}\left\| K \right\|_{\infty,M\times M}^{\frac{m}{n}}d_m^{\frac{n-m}{n}}
\end{equation*}
for all $n\in\mathbb{N}$. Using the monotonicity of $\left\{ \sigma_{n}\mid n\in\mathbb{N} \right\}$, by setting $m=\lfloor n/2 \rfloor$ we have the even more compact inequality
\begin{equation}
\label{eq:reduced-basis-greedy-rate-double-index}
  \sigma_n\leq \sqrt{2} \left\| K \right\|_{\infty,M\times M}^{\frac{1}{2}}d^{\frac{1}{2}}_{\lfloor n/2 \rfloor}\qquad \textrm{for all $n\in\mathbb{N}$, $n\geq 2$.}
\end{equation}
If we have a bound for $\left\{ d_n\mid n\in\mathbb{N} \right\}$, inequality \cref{eq:reduced-basis-greedy-rate-double-index} can be directly invoked to establish a bound for $\left\{ \sigma_n\mid n\in\mathbb{N} \right\}$, at the expense of comparing $\sigma_n$ with $d_{2n}$; in the regime $n\rightarrow\infty$ we may even expect the same rate of convergence at the expense of a larger constant. We emphasize here that the definition of $\left\{ d_n\mid n\in\mathbb{N} \right\}$ only involves elements in a compact subset $\mathscr{F}$ of the ambient Hilbert space $\mathscr{H}_K$; in our setting, the compact subset \cref{eq:compact-reduced-basis-set-our-setting} consists of only functions of the form $K \left( \cdot,x \right)$ for some $x\in M$, thus
\begin{equation}
  \label{eq:kolmogorov-width-gpr}
  \begin{aligned}
    d_n&=\inf_{x_1,\cdots,x_n\in M}\sup_{x\in M}\mathrm{dist}\left( K \left( \cdot,x \right),\mathrm{span}\left\{ K \left( \cdot,x_i \right)\mid 1\leq i\leq n \right\} \right)\\
    &=\inf_{x_1,\cdots,x_n\in M} \sup_{x\in M} \left[K \left( x,x \right)-K \left( x,x_n^1 \right)^{\top}
    K_{n,n}^{-1}
    K \left( x,x_n^1 \right)\right].
  \end{aligned}
\end{equation}
To ease notation, we will always denote $\mathscr{X}_n:=\left\{ x_1,\cdots,x_n \right\}$ as in \cref{lem:basic-project-estimates}. Write the maximum value of the function $P_{K,\mathscr{X}_n}$ over $M$ as
\begin{equation}
\label{eq:defn-power-function}
\begin{aligned}
  % P_{K,\mathscr{X}_n}\left( x \right)&:=\left\| K \left( \cdot,x \right)-P_n \left( K \left( \cdot,x \right) \right) \right\|_{\mathscr{H}_K}^2=\left[\mathrm{dist}\left( K \left( \cdot,x \right),\mathrm{span}\left\{ K \left( \cdot,x_i \right)\mid 1\leq i\leq n \right\} \right)\right]^2\\
  %   &=K \left( x,x \right)-\begin{pmatrix}
  %     K \left( x,x_1 \right),\cdots,K \left( x,x_n \right)
  %   \end{pmatrix}
  %   \begin{pmatrix}
  %     K \left( x_1,x_1 \right)&\cdots &K \left( x_1,x_n \right)\\
  %     \vdots & & \vdots\\
  %     K \left( x_n,x_1 \right)&\cdots&K \left( x_n,x_n \right)
  %   \end{pmatrix}^{-1}
  %   \begin{pmatrix}
  %     K \left( x,x_1 \right)\\
  %     \vdots\\
  %     K \left( x,x_n \right)
  %   \end{pmatrix},\\
  \Pi_{K,\mathscr{X}_n}&:=\sup_{x\in M}P_{K,\mathscr{X}_n}\left( x \right).
\end{aligned}
\end{equation}
The Kolmogorov width $d_n$ can be put in these notations as
\begin{equation}
\label{eq:kolmogorov-width-as-inf-power-function}
  d_n=\inf_{x_1,\cdots,x_n\in M}\Pi_{K,\mathscr{X}_n}.
\end{equation}
The problem of bounding $\left\{ d_n\mid n\in\mathbb{N} \right\}$ thus reduces to bounding the infimum of $\Pi_{K,\mathscr{X}_n}$ over all $n$-dimensional linear subspaces of $\mathscr{F}$.

%\TG{Very dangerous after this point: the standard use of fill distance to bound the power function require the domain to be open. When we naively extend the results to the manifold by extending/restricting functions from the manifold to/from a tubular neighborhood, first of all the fill distance does not go to zero if the samples only lie on the manifold.}

When $M$ is an open, bounded subset of a standard Euclidean space,
% or in our setting, when the Riemannian manifold is isometrically embedded into an ambient Euclidean space of larger dimension,
upper bounds for $\Pi_{K,\mathscr{X}_n}$ are often established---in a kernel-adaptive fashion---using the \emph{fill distance} \cite[Chapter 11]{Wendland2004}
\begin{equation}
  \label{eq:defn-fill-distance}
  h_{\mathscr{X}_n}:=\sup_{x\in M}\min_{x_j\in\mathscr{X}_n}\left\| x-x_j \right\|
\end{equation}
where $\left\| \cdot \right\|$ is the Euclidean norm of the ambient space. For instance, when $K$ is a squared exponential kernel \cref{eq:sq-exp-kernel-submfld} and the domain is a cube (or more generally, the domain should at least be compact and convex, as pointed out in \cite[Theorem 1]{WTW2017}) in a Euclidean space, \cite[Theorem 11.22]{Wendland2004} asserts that
\begin{equation}
  \label{eq:wendland-squared-exponential-kernel-estimates}
  \Pi_{K,\mathscr{X}_n}\leq \exp\left[c \frac{\log h_{\mathscr{X}_n}}{h_{\mathscr{X}_n}}\right]\qquad\forall h_{\mathscr{X}_n}\leq h_0
\end{equation}
for some constants $c>0$, $h_0>0$ depending only on $M$ and the kernel bandwidth $t>0$ in \cref{eq:sq-exp-kernel-submfld}. Similar bounds have been established in \cite{WuSchaback1993} for Mat{\'e}rn kernels, but the convergence rate is only polynomial. In this case, by the monotonicity of the function $x\mapsto \log x/x$ for $x\in \left( 0,e \right)$, we have, for all sufficiently small $h_{\mathscr{X}_n}$,
\begin{equation*}
  d_n=\inf_{x_1,\cdots,x_n\in M}\Pi_{K,\mathscr{X}_n}\leq\exp \left[ c \frac{\log h_n}{h_n} \right]
\end{equation*}
where
\begin{equation}
\label{eq:defn-min-fill-dist}
  h_n:=\inf_{\mathscr{X}_n\subset M,\,\,\left| \mathscr{X}_n \right|=n}h_{\mathscr{X}_n}
\end{equation}
is the minimum fill distance attainable for any $n$ sample points on $M$. We thus have the following theorem for the convergence rate of \cref{alg:gaussian-process-landmarking} for any compact, convex set in a Euclidean space:
\begin{theorem}
  \label{thm:greedy-conditional-covariance-decay-convex}
Let $\Omega\subset \mathbb{R}^D$ be a compact and convex subset of the $D$-dimensional Euclidean space, and consider a Gaussian process $\GP \left( m,K \right)$ defined on $\Omega$, with the covariance function $K$ being of the squared exponential form \eqref{eq:sq-exp-kernel-submfld} with respect to the ambient $D$-dimensional Euclidean distance. Let $X_1,X_2,\cdots,$ denote the sequence of landmarks greedily picked on $\Omega$ according to \cref{alg:gaussian-process-landmarking}, and define for any $n\in\mathbb{N}$ the maximum MSPE on $\Omega$ with respect to the first $n$ landmarks $X_1,\cdots,X_n$ as
%Let $\Omega\subset \mathbb{R}^D$ be a compact and convex subset of the $D$-dimensional Euclidean space, and consider a Gaussian process $\GP \left( m,K \right)$ defined on $\Omega$, with the covariance function $K$ being of the squared exponential form \cref{eq:sq-exp-kernel-submfld} with respect to the ambient $D$-dimensional Euclidean distance. Let $X_1,X_2,\cdots,$ denote the sequence of landmarks greedily picked on $\Omega$ according to \cref{alg:gaussian-process-landmarking}, and define for any $n\in\mathbb{N}$ the maximum conditional covariance on $\Omega$ with respect to the first $n$ landmarks $X_1,\cdots,X_n$ as
\begin{equation*}
  \sigma_n=\max_{x\in \Omega}\left[ K \left( x,x \right)-K \left( x,X^1_n \right)^{\top}K_n^{-1} K \left( x,X_n^1 \right)\right]
\end{equation*}
where the notations $K \left( x,X_n^1 \right)$ and $K_n$ are defined in \Cref{sec:gauss-proc-active}.
% \begin{equation*}
%   \sigma_n=\max_{x\in \Omega}\left[ K \left( x,x \right)-\begin{pmatrix}
%       K \left( x,X_1 \right),\cdots,K \left( x,X_n \right)
%     \end{pmatrix}
%     \begin{pmatrix}
%       K \left( X_1,X_1 \right)&\cdots &K \left( X_1,X_n \right)\\
%       \vdots & & \vdots\\
%       K \left( X_n,X_1 \right)&\cdots&K \left( X_n,X_n \right)
%     \end{pmatrix}^{-1}
%     \begin{pmatrix}
%       K \left( x,X_1 \right)\\
%       \vdots\\
%       K \left( x,X_n \right)
%     \end{pmatrix} \right].
% \end{equation*}
Then
\begin{equation}
\label{eq:convergence-rate-in-fill-distance}
  \sigma_n=O \left( \beta^{\frac{\log h_{\lfloor n/2\rfloor}}{h_{\lfloor n/2\rfloor}}} \right)\qquad\textrm{as $n\rightarrow\infty$}%\textrm{as $h_n\rightarrow0$}
\end{equation}
for some positive constant $\beta>1$ depending only on the geometry of the domain $\Omega$ and the bandwidth of the squared exponential kernel $K$; $h_n$ is the minimum fill distance of $n$ arbitrary points on $\Omega$ (c.f. \cref{eq:defn-min-fill-dist}).
\end{theorem}

\begin{proof}
  By the monotonicity of the sequence $\left\{ \sigma_n\mid n\in\mathbb{N} \right\}$, it suffices to establish the convergence rate for a subsequence. Using directly \cref{eq:reduced-basis-greedy-rate-double-index}, \cref{eq:kolmogorov-width-as-inf-power-function}, \cref{eq:wendland-squared-exponential-kernel-estimates}, and the definition of $h_n$ in \cref{eq:defn-min-fill-dist}, we have the inequality for all $\mathbb{N}\ni n\geq N$:
\begin{equation*}
    \sigma_{2n}\leq \sqrt{2}\left\| K \right\|^{\frac{1}{2}}_{\infty,\,\Omega\times \Omega}\exp \left[ \frac{c}{2}\frac{\log h_n}{h_n} \right]=\sqrt{2}\left\| K \right\|^{\frac{1}{2}}_{\infty,\,\Omega\times \Omega}\beta^{\frac{\log h_n}{h_n}}
  \end{equation*}
where $\beta:=\exp\left(c/2\right)>1$. Here the positive constants $N=N \left( \Omega,t \right)>0$ and $c=c \left( \Omega,t \right)>0$ depend only on the geometry of $\Omega$ and the bandwidth of the squared exponential kernel. This completes the proof.
\end{proof}

Convex bodies in $\mathbb{R}^D$ are far too restricted as a class of geometric objects for modeling anatomical surfaces in our main application \cite{GPLMK2}. The rest of this section will be devoted to generalizing the convergence rate for squared exponential kernels \cref{eq:sq-exp-kernel-submfld} to their reweighted counterparts \cref{eq:sq-exp-kernel-submfld-reweighted}, and more importantly, for submanifolds of the Euclidean space. The crucial ingredient is an estimate of the type \cref{eq:wendland-squared-exponential-kernel-estimates} bounding the sup-norm of the squared power function using fill distances, tailored for restrictions of the squared exponential kernel 
\begin{equation}
\label{eq:squared-exponential-kernel-continuous-form}
  K_{\epsilon} \left( x,y \right)=\exp \left( -\frac{1}{2\epsilon}\left\| x-y \right\|^2 \right),\quad x,y\in M
\end{equation}
as well as the reweighted version
\begin{equation}
\label{eq:reweighted-kernel-continuous-form}
  K_{\epsilon}^w \left( x,y \right)=\int_M w \left( z \right) \exp \left[ -\frac{1}{2\epsilon}\left( \left\| x-z \right\|^2+\left\| z-y \right\|^2 \right) \right]\dd \mathrm{vol}_M \left( z \right),\quad x,y\in M
\end{equation}
where $w:M\rightarrow\mathbb{R}_{\geq 0}$ is a non-negative weight function. Note that when $w \left( x \right)\equiv 1$, $\forall x\in M$ the reweighted kernel \cref{eq:reweighted-kernel-continuous-form} does not coincide with the squared exponential kernel \cref{eq:squared-exponential-kernel-continuous-form}, not even up to normalization, since the domain of integration is $M$ instead of the entire $\mathbb{R}^D$; neither does na{\"i}vely enclosing the compact manifold $M$ with a compact, convex subset $\Omega$ of the ambient space and reusing \cref{thm:greedy-conditional-covariance-decay-convex} by extending/restricting functions to/from $M$ to $\Omega$ seem to work, since the samples are constrained to lie on $M$ but the convergence will be in terms of fill distances in $\Omega$. Nevertheless, the desired bound can be established using local parametrizations of the manifold, i.e., working within each local Lipschitz coordinate chart and taking advantage of the compactness of $M$.

We will henceforth impose no additional assumptions, other than compactness and smoothness, on the geometry of the Riemannian manifold $M$. In the first step we refer to a known uniform estimate from \cite[Theorem 17.21]{Wendland2004} for power functions on a compact Riemannian manifold.

\begin{lemma}
\label{lem:bounding-power-function-reweighted}
Let $M$ be a $d$-dimensional $C^{\ell}$ compact manifold isometrically embedded in $\mathbb{R}^D$ (where $D>d$), and let $\Phi\in C^{2k}\left( M\times M \right)$ be any positive definite kernel function on $M\times M$ with $2k\leq \ell$.
%either the squared exponential kernel function $K_{\epsilon}$ in \cref{eq:squared-exponential-kernel-continuous-form} or its reweighted version $K_{\epsilon}^w$ in \cref{eq:reweighted-kernel-continuous-form}.
There exists a positive constant $h_0=h_0 \left( M \right)>0$ depending only on the geometry of the manifold $M$ such that, for any collection of $n$ distinct points $\mathscr{X}_n=\left\{ x_1,\cdots,x_n \right\}$ on $M$ with $h_{\mathscr{X}_n}\leq h_0$, the following inequality holds:
\begin{equation*}
  \Pi_{\Phi,\mathscr{X}_n}=\sup_{x\in M}P_{\Phi,\mathscr{X}_n} \left( x \right)\leq C h_{\mathscr{X}_n}^{2k}
\end{equation*}
where $C=C \left( k,M,\Phi \right)>0$ is a positive constant depending only on the manifold $M$ and the kernel function $\Phi$. This of course further implies for all $h_n\leq h_0$
\begin{equation*}
  \inf_{\mathscr{X}_n\subset M,\,\left| \mathscr{X}_n \right|=n}\Pi_{\Phi,\mathscr{X}_n}\leq C h_n^{2k}
\end{equation*}
where $h_n$ is the minimum fill distance of $n$ arbitrary points on $\Omega$ (c.f. \cref{eq:defn-min-fill-dist}).
\end{lemma}

\begin{proof}
  This is essentially \cite[Theorem 17.21]{Wendland2004}, with the only adaptation that the definition of the power function throughout \cite{Wendland2004} is the square root of the $P_{\Phi,\mathscr{X}_n}$ in our definition \cref{eq:defn-power-function}.
\end{proof}

\cref{lem:bounding-power-function-reweighted} suggests that the convergence of \cref{alg:gaussian-process-landmarking} is faster than any polynomial of $h_n$. The dependence on $h_n$ can be made more direct in terms of the number of samples $n$ by the following geometric lemma.

\begin{lemma}
\label{lem:fill-to-packing}
  Let $M$ be a $d$-dimensional $C^{\ell}$ compact Riemannian manifold isometrically embedded in $\mathbb{R}^D$ (where $D>d$). Denote $\omega_{d-1}$ for the surface measure of the unit sphere in $\mathbb{R}^d$, and $\mathrm{Vol} \left( M \right)$ for the volume of $M$ induced by the Riemannian metric. There exists a positive constant $N=N \left( M \right)>0$ depending only on the manifold $M$ such that
\begin{equation*}
  h_n\leq \left( \frac{2^{d+1} d}{\omega_{d-1}}\mathrm{Vol}\left( M \right) \right)^{\frac{1}{d}}\cdot n^{-\frac{1}{d}}\qquad\textrm{for any $\mathbb{N}\ni n\geq N$.}
\end{equation*}
\end{lemma}

\begin{proof}
  For any $r>0$ and $x\in M$, we denote $B^{D}_r \left( x \right)$ for the (extrinsic) $D$-dimensional Euclidean ball centered at $x\in M$, and set $B_r \left( x \right):=B^D_r \left( x \right)\cap M$. In other words, $B_r \left( x \right)$ is a ball of radius $r$ centered at $x\in M$ with respect to the ``chordal'' metric on $M$ induced from the ambient Euclidean space $\mathbb{R}^D$. Define the \emph{covering number} and the \emph{packing number} for $M$ with respect to the chordal metric balls by
\begin{equation*}
  \begin{aligned}
    \mathscr{N}\left( r \right)&:=\mathscr{N}\left( M,\left\| \cdot \right\|_D,r \right)\\
    &:=\min_{n\in\mathbb{N}}\left\{ M\subset \bigcup_{i=1}^n B_r \left( x_i \right)\mid x_i\in M, 1\leq i\leq n \right\},\\
    \mathscr{P}\left( r \right)&:=\mathscr{P}\left( M,\left\| \cdot \right\|_D,r \right)\\
    &:=\max_{n\in\mathbb{N}}\Bigg\{ \bigcup_{i=1}^n B_{r/2} \left( x_i \right)\subset M, B_{r/2} \left( x_i \right)\cap B_{r/2} \left( x_j \right)=\emptyset\\
    &\hspace{1in}\textrm{for all } 1\leq i\neq j\leq n\,\Big|\,x_i\in M,1\leq i\leq n \Bigg\}.
  \end{aligned}
\end{equation*}
By the definition of fill distance and $h_n$ (c.f. \cref{eq:defn-min-fill-dist}), the covering number $\mathscr{N}\left( h_n \right)$ is lower bounded by $n$; furthermore, by the straightforward inequality $\mathscr{P}\left( r \right)\geq \mathscr{N}\left( r \right)$ for all $r>0$, we have
\begin{equation*}
  n<\mathscr{N} \left( h_n \right)\leq \mathscr{P}\left( h_n \right),
\end{equation*}
i.e. there exist a collection of $n$ points $x_1,\cdots,x_n\in M$ such that the $n$ chordal metric balls $\left\{ B_{h_n/2}\left( x_i \right)\mid 1\leq i\leq n \right\}$ form a packing of $M$.  Thus
\begin{equation*}
  \sum_{i=1}^n\mathrm{Vol} \left( B_{h_n/2}\left( x_i \right) \right)\leq \mathrm{Vol}\left( M \right)<\infty
\end{equation*}
where the last inequality follows from the compactness of $M$. The volume of each $B_{h_n/2}\left( x_i \right)$ can be expanded asymptotically for small $h_n$ as (c.f. \cite{KarpPinsky1989})
\begin{equation}
\label{eq:packing-volume}
  \mathrm{Vol}\left( B_{h_n/2} \left( x \right) \right)=\frac{\omega_{d-1}}{d} \left( \frac{h_n}{2} \right)^d \left[ 1+\frac{2 \left\| B \right\|_x^2-\left\| H \right\|_x^2}{8 \left( d+2 \right)}\left( \frac{h_n}{2} \right)^2 \right]+O \left( h_n^{d+3} \right)\,\,\textrm{as $h_n\rightarrow 0$}
\end{equation}
where $\omega_{d-1}$ is the surface measure of the unit sphere in $\mathbb{R}^d$, $B$ is the second fundamental form of $M$, and $H$ is the mean curvature normal. The compactness of $M$ ensures the boundedness of all these extrinsic curvature terms. Pick $n$ sufficiently large so that $h_n$ is sufficiently small (again by the compactness of $M$) to ensure
\begin{equation*}
  \mathrm{Vol}\left( B_{h_n/2} \left( x \right) \right)\geq \frac{\omega_{d-1}}{2d} \left( \frac{h_n}{2} \right)^d.
\end{equation*}
It then follows from \cref{eq:packing-volume} that
\begin{equation*}
  \frac{n\omega_{d-1}}{2d}\left( \frac{h_n}{2} \right)^d\leq \mathrm{Vol}\left( M \right)\quad\Rightarrow\quad h_n\leq \left( \frac{2^{d+1}d}{\omega_{d-1}}\mathrm{Vol}\left( M \right) \right)^{\frac{1}{d}}\cdot n^{-\frac{1}{d}}.
\end{equation*}
\end{proof}

We are now ready to conclude that \cref{alg:gaussian-process-landmarking} converges faster than any inverse polynomial in the number of samples with our specific choice of kernel functions, regardless of the presence of reweighting.

\begin{theorem}
\label{thm:greedy-conditional-covariance-decay}
Let $M$ be a $d$-dimensional $C^{\infty}$ compact manifold isometrically embedded in $\mathbb{R}^D$ (where $D>d$), and let $\Phi\in C^{\infty}\left( M\times M \right)$ be any positive definite kernel function on $M$. For any $k\in\mathbb{N}$, there exist positive constants $N=N \left( M \right)>0$ and $C_k=C_k \left( M,\Phi \right)>0$ such that
\begin{equation*}
  \sigma_n\leq C_k n^{-\frac{k}{d}}\quad\textrm{for all $n\geq N$}.
\end{equation*}
Equivalently speaking, \cref{alg:gaussian-process-landmarking} converges at rate $O \left( n^{-\frac{k}{d}} \right)$ for all $k\in\mathbb{N}$, but with constants possibly depending on $k$.
\end{theorem}
\begin{proof}
  Use \cref{lem:bounding-power-function-reweighted}, \cref{lem:fill-to-packing} and the regularity of the kernel function $\Phi$.
\end{proof}

It is natural to conjecture that a faster rate of convergence than the conclusion of \cref{thm:greedy-conditional-covariance-decay}, for instance exponential rate of convergence, should hold for the reweighted kernel \eqref{eq:reweighted-kernel-continuous-form}, or at least for the Euclidean radial basis kernel \cref{eq:squared-exponential-kernel-continuous-form}; this can be empirically validated with numerical experiments, see e.g. the log-log plots in \cref{fig:loglog-conv-rate} depicting the decay of MSPE (i.e., $\sigma_n^2$) with respect to the increasing number of landmarks (i.e., $n$). Unfortunately, \cref{thm:greedy-conditional-covariance-decay} is about as far as we can get with our current techniques, unless we impose additional assumptions on the regularity of the manifolds of interest. It is tempting to proceed directly as in \cite[Theorem 17.21]{Wendland2004} by working locally on coordinate charts and citing the exponential convergence result for radial basis kernels in \cite[Theorem 11.22]{Wendland2004}; unfortunately, even though kernel $K_{\epsilon}$ is of radial basis type in the ambient space $\mathbb{R}^D$, it is generally no longer of radial basis type in local coordinate charts, unless one imposes additional restrictive assumptions on the growth of the derivatives of local parametrization maps (e.g. all coordinate maps are affine). We will not pursue the theoretical aspects of these additional assumptions in this paper.

\begin{figure}
\subfloat[Reweighted kernel \eqref{eq:reweighted-kernel-continuous-form}]{\includegraphics[width=0.5\textwidth]{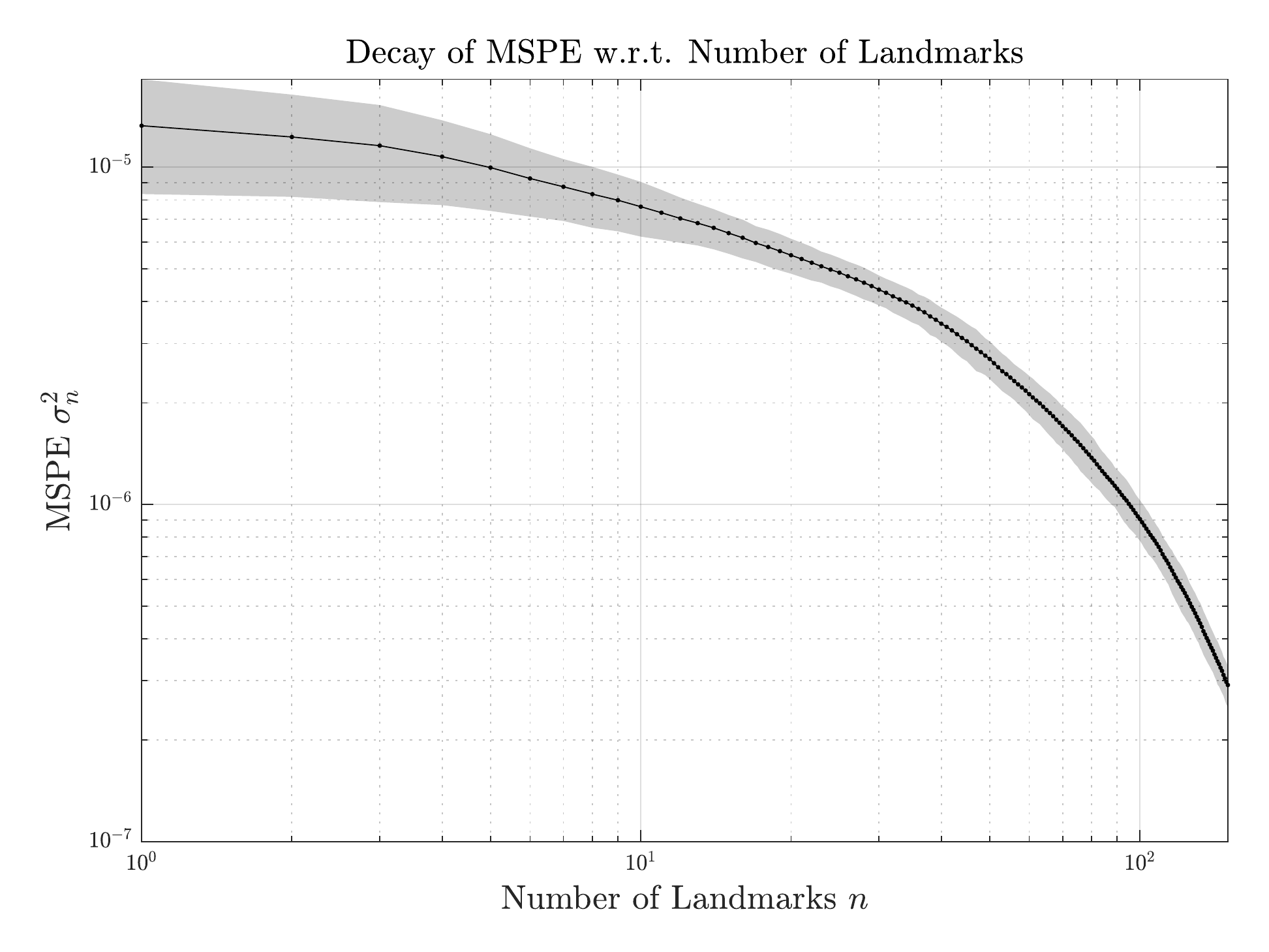}}
\subfloat[Euclidean kernel \eqref{eq:squared-exponential-kernel-continuous-form}]{\includegraphics[width= 0.5\textwidth]{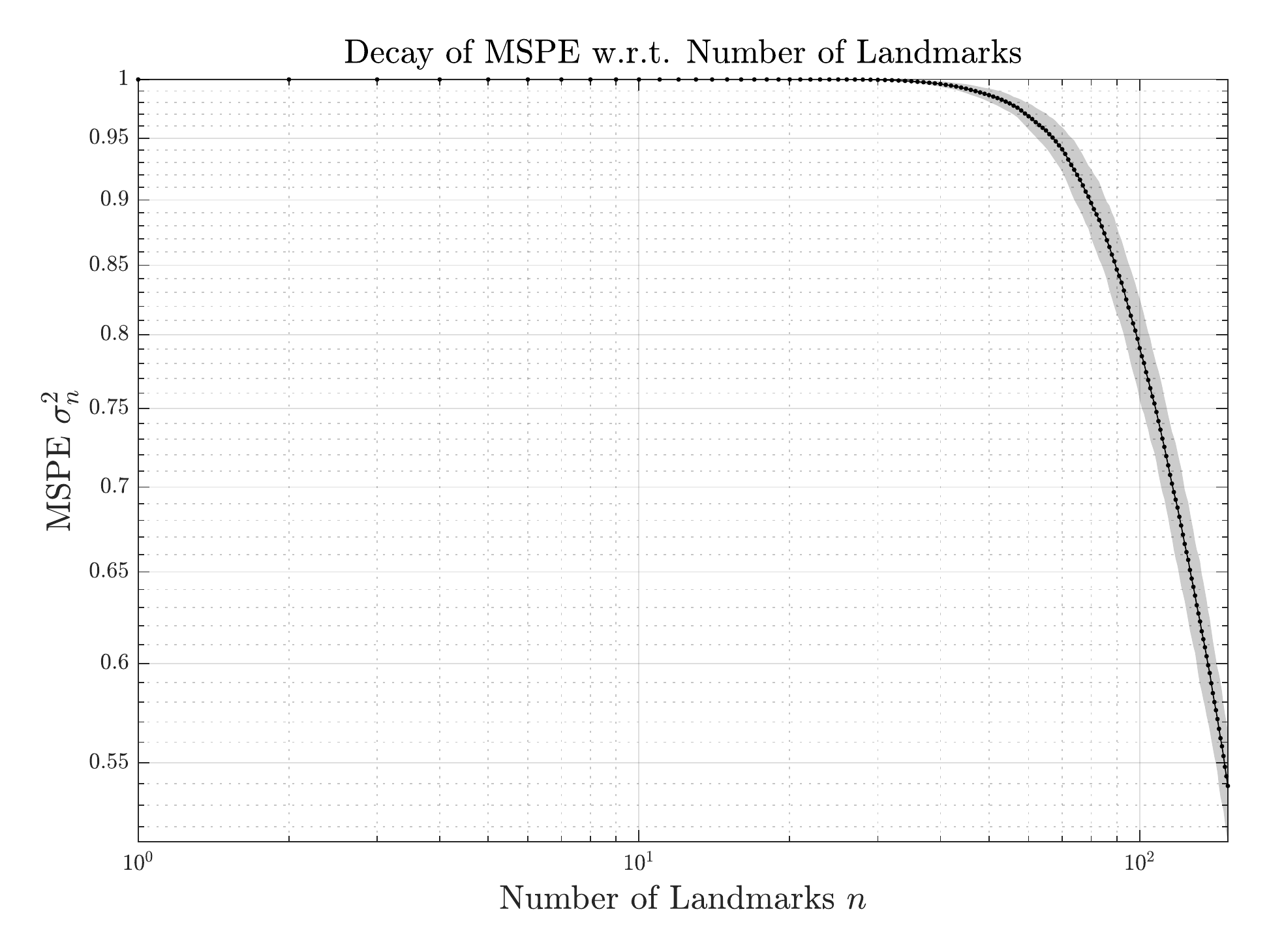}}
\caption{Log-log plots illustrating the convergence of MSPE with respect to the number of Gaussian process landmarks produced using the reweighted kernel \eqref{eq:reweighted-kernel-continuous-form} or the Euclidean kernel \eqref{eq:squared-exponential-kernel-continuous-form}, for a collection of $116$ second mandibular molars of prosimian primates and closely related non-primates \cite{PNAS2011,GPLMK2}. Each point on either curve is obtained by averaging the $116$ MSPE over the entire dataset, and the transparent bands represent pointwise confidence bands of two standard deviations. For both plots we vary the number of landmarks from $1$ to $150$; the total number of vertices on each of the $116$ triangular meshes vary around $5000$. For both plots, the MSPE decays linearly for sufficiently large $n$ on a log-log scale, suggesting exponential convergence with respect to the number of Gaussian process landmarks.}
\label{fig:loglog-conv-rate}
\end{figure}

\begin{remark}
The asymptotic optimality of the rate established in \cref{thm:greedy-conditional-covariance-decay} for Gaussian process landmarking follows from \cref{eq:greedy-reduced-basis-theorem}. In other words, the Gaussian process landmarking algorithm leads to a rate of decay of the $\infty$-norm of the pointwise MSPE that is at least as fast as any other landmarking algorithms, including random or uniform sampling on the manifold. In our application of comparative biology that motivated this paper, it is more important that Gaussian process landmarking is capable of identifying biologically meaningful and operationally homologous points across the anatomical surfaces even when the number of landmarks is not large ($n\ll \infty$); see \cite{GPLMK2} for more details. A more thorough theory explaining this advantageous aspect of Gaussian process landmarking will be left for future work.
%The asymptotic optimality of the rate established in \cref{thm:greedy-conditional-covariance-decay} for Gaussian process landmarking follows from \cref{eq:greedy-reduced-basis-theorem}. In other words, the Gaussian process landmarking algorithm leads to a rate of decay of the $\infty$-norm of conditional covariances that is at least as fast as any other landmarking algorithms, including random or uniform sampling on the manifold. In our application of comparative biology that motivated this paper, it is more important that Gaussian process landmarking is capable of identifying biologically meaningful and operationally homologous points across the anatomical surfaces even when the number of landmarks in small ($n\ll \infty$); see \cite{GPLMK2} for more details. A more thorough theory explaining this advantageous aspect of Gaussian process landmarking will be left for future work.
\end{remark}

\section{Discussion and Future Work}
\label{sec:disc-future-work}

This paper discusses a greedy algorithm for automatically selecting representative points on compact manifolds, motivated by the methodology of experimental design with Gaussian process prior in statistics. With a carefully modified heat kernel specified as the covariance function in the Gaussian process prior, our algorithm is capable of producing biologically highly meaningful feature points on some anatomical surfaces. Application of this landmarking scheme for real anatomical datasets is detailed in a companion paper \cite{GPLMK2}.

A future direction of interest is to build theoretical analysis for the optimal experimental design aspects of manifold learning: Whereas existing manifold learning algorithms estimate the underlying manifold \emph{from} discrete samples, our algorithm concerns economical strategies for encoding geometric information \emph{into} discrete samples. The landmarking procedure can also be interpreted as a compression scheme for manifolds; correspondingly, standard manifold learning algorithms may be understood as a decoding mechanism. Our theory is also of potential interest to adaptive matrix sensing and image completion problems, in which sensing procedures and subsampling schemes can be designed to collect more information for the ease of reconstruction. Some related work of this type include \cite{AB2013,WZ2013,WYS2017} and the references therein.

This work stems from an attempt to impose Gaussian process priors on diffeomorphisms between distinct but comparable biological structures, with which a rigorous Bayesian statistical framework for biological surface registration may be developed. The motivation is to measure the uncertainty of pairwise bijective correspondences automatically computed from geometry processing and computer vision techniques. We hope this MSPE-based sequential landmarking algorithm will shed light upon generalizing covariance structures from a single shape to pairs or even collections of shapes for collection shape analysis.

\section*{Acknowledgments}
TG would like to thank Peng Chen (UT Austin) for pointers to the reduced basis method literature, Chen-Yun Lin (Duke) for many useful discussions on heat kernel estimates, Daniel Sanz-Alonso (University of Chicago) for general discussion on Gaussian processes and RKHS, and Yingzhou Li (Duke) and Jianfeng Lu (Duke) for discussions on the numerical linear algebra perspective. The authors would also like to thank Shaobo Han, Rui Tuo, Sayan Mukherjee, Robert Ravier, Shan Shan, and Albert Cohen for many inspirational discussions, as well as Ethan Fulwood, Bernadette Perchalski, Julia Winchester, and Arianna Harrington for assistance with collecting observer landmarks. Last but not the least, we sincerely appreciate the constructive feedback from the anonymous reviewers.

\appendix

\section{Reproducing Kernel Hilbert Spaces}
\label{sec:repr-kern-hilb}

For any positive semi-definite symmetric kernel function $K:M\times M\rightarrow \mathbb{R}$ defined on a complete metric measure space $M$, Mercer's Theorem \cite[Theorem 3.6]{CS2000} states that $K$ admits a uniformly convergent expansion of the form
\begin{equation*}
  K \left( x,y \right)=\sum_{i=0}^{\infty}e^{-\lambda_i}\phi_i \left( x \right)\phi_i \left( y \right),\quad\forall x,y\in M,
\end{equation*}
where $\left\{\phi_i\right\}_{i=0}^{\infty}\subset L^2\left( M \right)$ are the eigenfunctions of the integral operator
\begin{equation*}
  \begin{aligned}
    T_K:L^2 \left( M \right)&\rightarrow L^2 \left( M \right)\\
    T_Kf \left( x \right)&:= \int_M K \left( x,y \right)f \left( y \right)\dd\mathrm{vol}_M \left( y \right),\quad\forall f\in L^2 \left( M \right)
  \end{aligned}
\end{equation*}
and $e^{-\lambda_i}, i=0,1,\cdots$, ordered so that $e^{-\lambda_0}\geq e^{-\lambda_1}\geq e^{-\lambda_2}\geq\cdots$, are the eigenvalues of this integral operator corresponding to the eigenfunctions $\phi_i, i=0,1,\cdots$, respectively. Regression under this framework amounts to restricting the regression function to lie in the Hilbert space
\begin{equation}
\label{eq:defn-rkhs-intro}
  \mathscr{H}_K:= \left\{ f=\sum_{i=0}^{\infty}\alpha_i\phi_i \,\bigg|\, \alpha_i\in\mathbb{R}, \displaystyle\sum_{i=0}^{\infty}e^{\lambda_i}\alpha_i^2<\infty \right\}
\end{equation}
on which the inner product is defined as
\begin{equation}
\label{eq:rkhs-inner-product}
  \left\langle f,g \right\rangle_{\mathscr{H}_K}=\sum_{i=0}^{\infty}e^{\lambda_i} \left\langle f,\phi_i \right\rangle_{L^2\left(M\right)} \left\langle g,\phi_i \right\rangle_{L^2 \left( M \right)}.
\end{equation}
The reproducing property is reflected in the identity
\begin{equation}
\label{eq:rkhs-reproducing-property-kernel}
  \left\langle K \left( \cdot,x \right),K \left( \cdot,y \right) \right\rangle_{\mathscr{H}_K} = K \left( x,y \right)\qquad\forall x,y\in M.
\end{equation}
Borrowing terminologies from kernel-based learning methods (see e.g. \cite{CS2000} and \cite{SS2001}), the eigenfunctions and eigenvalues of $T_K$ define a \emph{feature mapping}
\begin{equation}
\label{eq:feature-mapping}
  M\ni x\longmapsto \Phi \left( x \right):= \left( e^{-\lambda_0/2}\phi_0 \left( x \right),e^{-\lambda_1/2}\phi_1 \left( x \right),\cdots,e^{-\lambda_i/2}\phi_i \left( x \right),\cdots \right)\in\ell^2
\end{equation}
such that the kernel value $K \left( x,y \right)$ at an arbitrary pair $x,y\in M$ is given exactly by the inner product of $\Phi \left( x \right)$ and $\Phi \left( y \right)$ in the feature space $\ell^2$, i.e.
\begin{equation*}
  K \left( x,y \right) = \left\langle \Phi \left( x \right),\Phi \left( y \right) \right\rangle_{\ell^2},\quad\forall x,y\in M.
\end{equation*}
This interpretation leads to the following equivalent form of the RKHS \eqref{eq:defn-rkhs-intro}
\begin{equation}
\label{eq:defn-rkhs-intro-reinterp}
\begin{aligned}
  \mathscr{H}_K &= \left\{ f=\displaystyle \sum_{i=0}^{\infty}\beta_i \cdot e^{-\lambda_i/2}\phi_i= \left\langle \beta,\Phi \right\rangle_{\ell^2} \,\bigg|\,\beta=\left( \beta_0,\beta_1,\cdots,\beta_i,\cdots \right)\in\ell^2 \right\}\\
%  &=\overline{\mathrm{span}\left\{ \sum_{i\in I}a_i K \left( \cdot,x_i \right)\mid a_i\in\mathbb{R},x_i\in M,\mathrm{card}\left( I \right)<\infty \right\}}\\
  &=\overline{\mathrm{span}\left\{ \sum_{i\in I}a_i K \left( \cdot,x_i \right)\mid a_i\in\mathbb{R},x_i\in M,\mathrm{card}\left( I \right)<\infty \right\}}.
\end{aligned}
\end{equation}
In other words, the RKHS framework embeds the Riemannian manifold $M$ into an infinite dimensional Hilbert space $\ell^2$, and converts the (generically) nonlinear regression problem on $M$ into a linear regression problem on a subset of $\ell^2$. We refer interested readers to \cite{BBG1994,JMS2008,Wu2017} for more discussions of this type of embedding in the nonlinear dimension reduction literature.

\bibliographystyle{siamplain}
\bibliography{references}
\end{document}